\pdfoutput=1

\documentclass[a4paper,11pt,sort]{elsarticle}
\usepackage{fullpage}

\journal{Journal of Computer and System Sciences}

\usepackage{latexsym}
\usepackage{amssymb} 
\usepackage{mathtools}

\usepackage{xcolor}
\definecolor{darkred}{rgb}{0.7,0,0}

\newcommand{\sts}{\textsc{Sequential Token Swapping}}
\newcommand{\ststs}[2]{(#1, #2)-\textsc{STS}}
\newcommand{\substs}{\textsc{Sub-STS}}
\newcommand{\lensubsts}{\mathsf{\lambda}}

\usepackage{hyperref}

\usepackage{ntheorem} % necessary to allow cleveref for automatically infer the type of theorem-like environment
\usepackage[capitalise,nameinlink]{cleveref}

\newtheorem{theorem}{Theorem}[section]

\newtheorem{lemma}[theorem]{Lemma}

\newtheorem{corollary}[theorem]{Corollary}
\newtheorem{claim}[theorem]{Claim}
\crefname{claim}{Claim}{Claims}

\newproof{proof}{Proof}

\let\doendproof\endproof
\renewcommand\endproof{~\hfill\qed\doendproof} % automatic qed (compatible with lineno)

\crefname{claim}{Claim}{Claims}

\begin{document}

\begin{frontmatter}

\title{Sequentially Swapping Tokens: Further on Graph Classes\tnoteref{t1,t2}}
\tnotetext[t1]{%
A preliminary version appeared in the proceedings of
the 48th International Conference on Current Trends in Theory and Practice of Computer Science (SOFSEM 2023),
Lecture Notes in Computer Science 13878 (2023) 222--235.
}
\tnotetext[t2]{%
Partially supported
by JSPS KAKENHI Grant Numbers 
JP17H01698, % Ono (Ono B) 
JP17K19960, % Ono (Kawamura Hoga) 
JP18H04091, % Otachi (Uehara A)
JP20H05793, % Otachi (Ito Henkaku B)
JP20H05967, % Ono (Makino Henkaku A)
JP21K11752, % Otachi (Otachi C)
JP21K19765, % Ono (Hoga)
JP21K21283, % Kiya (Kiya, Start-up)
JP22H00513. % Ono, Otachi (Ono A)
}

\author[1]{Hironori Kiya}
\ead{h-kiya@econ.kyushu-u.ac.jp}
\author[2]{Yuto Okada}
\ead{okada.yuto.b3@s.mail.nagoya-u.ac.jp}
\author[2]{Hirotaka Ono}
\ead{ono@nagoya-u.jp}
\author[2]{Yota Otachi\corref{c1}}
\ead{otachi@nagoya-u.jp}
\cortext[c1]{Corresponding author.}

\affiliation[1]{
organization={Kyushu University},
city={Fukuoka},
country={Japan}}

\affiliation[2]{
organization={Nagoya University},
city={Nagoya},
country={Japan}}

\begin{abstract}
We study the following variant of the 15 puzzle.
Given a graph and two token placements on the vertices,
we want to find a walk of the minimum length (if any exists) such that
the sequence of token swappings along the walk obtains one of the given token placements from the other one.
This problem was introduced as \textsc{Sequential Token Swapping} by Yamanaka et al.~[JGAA 2019],
who showed that the problem is intractable in general
but polynomial-time solvable for trees, complete graphs, and cycles.
In this paper, we present a polynomial-time algorithm for block-cactus graphs, which include all previously known cases.
We also present general tools for showing the hardness of the problem on restricted graph classes
such as chordal graphs and chordal bipartite graphs.
We also show that the problem is hard on grids and king's graphs,
which are the graphs corresponding to the 15 puzzle and its variant with relaxed moves.

\end{abstract}


\end{frontmatter}

%

% introduction
\section{Introduction}
\label{sec:intro}

Let $G = (V,E)$ be an undirected graph
and $f, f' \colon V \to \{1,\dots,c\}$ be colorings of $G$.\footnote{%
By a coloring, we mean a mapping from the vertex set to a color set, which is not necessarily a proper coloring.}
We call a sequence $\langle f_{1}, \dots, f_{p} \rangle$ of colorings of $G$
a \emph{swapping sequence} of length $p-1$ from $f$ to $f'$ if 
$f_{1} = f$, $f_{p} = f'$, and 
there is a walk 
$\langle w_{1}, w_{2}, \dots, w_{p} \rangle$ such that
for $2 \le i \le p$,
$f_{i}$ is obtained from $f_{i-1}$ by \emph{swapping} the colors of $w_{i-1}$ and $w_{i}$; that is,
$f_{i}(w_{i}) = f_{i-1}(w_{i-1})$, $f_{i}(w_{i-1}) = f_{i-1}(w_{i})$, and 
$f_{i}(v) = f_{i-1}(v)$ for $v \notin \{w_{i-1}, w_{i}\}$.
See \cref{fig:prob_eg}.
Now the problem can be formulated as follows.
\begin{description}
  \item[Problem:] \sts
  \item[Input:] A graph $G = (V, E)$, colorings $f, f'$ of $G$, and an integer $k$.
  \item[Question:] Is there a swapping sequence of length at most $k$ from $f$ to $f'$?
\end{description}
We assume that $f$ and $f'$ color the same number of vertices for each color
since otherwise it becomes a trivial no-instance.
We also assume that the input graph $G$ is connected
as a swapping sequence affects only one connected component.

\begin{figure}[tbh]
  \centering
  \includegraphics[width=\textwidth]{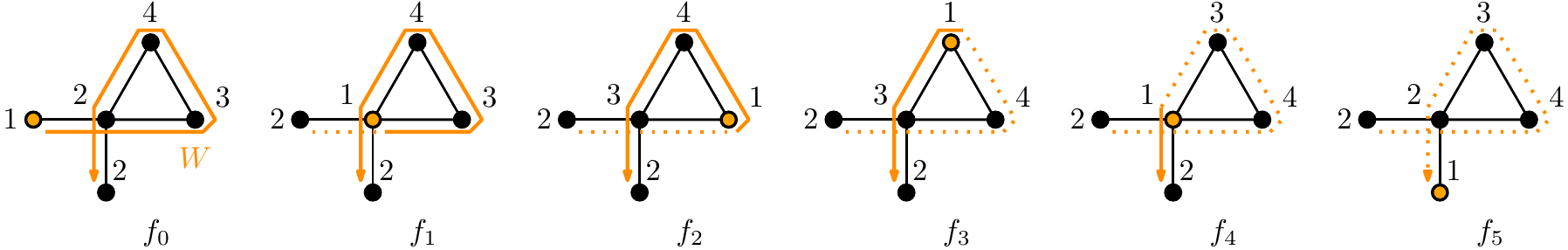}
  \caption{An example of a swapping sequence.}
  \label{fig:prob_eg}
\end{figure}

The intuition behind its name, \sts, is as follows:
we consider a coloring as an assignment of colored tokens (or pebbles) to the vertices;
we proceed along a walk; and
when we visit an edge in the walk, we swap the tokens on the endpoints.
For the ease of presentation, we often use the concept of tokens in this paper.
For example, we call the token on the first vertex of the walk the \emph{moving token}
as it will always be the one exchanged during the swapping sequence.
In other words, $f_{i}(w_{i}) = f_{1}(w_{1})$ holds for all $i$.

Yamanaka et al.~\cite{YamanakaDHKNOSS19} introduced \sts{}
as a variant of the (generalized) 15 puzzle~(\cref{fig:15-puzz}), in which the first and last vertices in a swapping sequence are given as part of input.
They showed that \sts{} is polynomial-time solvable in some restricted cases such as trees, complete graphs, and cycles.
They also showed that there is a constant $\varepsilon > 0$ such that the shortest length of a swapping sequence is NP-hard
to approximate within a factor $1 + \varepsilon$.

\paragraph{Our results}

We unify and extend the positive results in~\cite{YamanakaDHKNOSS19} by showing that 
\sts{} is polynomial-time solvable on block-cactus graphs,
which include the classes of trees, complete graphs, and cycles.
To this end, we first show that \sts{} on a graph is reducible 
to a generalized problem (called \substs) on its biconnected components, which may be of independent interest.
We then show that the generalized problem \substs{} can be solved in polynomial time on complete graphs and cycles.
As a byproduct, we also show that the generalized 15 puzzle is polynomial-time solvable on the same graph class.

To complement the positive results, 
we show negative results on several classes of graphs.
We first present two general tools for showing the NP-hardness of \sts{} on restricted graph classes.
One is for the \emph{few-color case}, where we use only a fixed number of colors,
and the other is for the \emph{colorful case}, where we use a unique color for each vertex.
The graph classes covered by the general tools include chordal graphs and chordal bipartite graphs.
We also show the hardness on grids and king's graphs that play important roles in the connection to puzzles~\cite{JohnsonS1879} and 
video games~\cite{Puzzle-and-Dragons_web}.
For them, our general tools cannot be applied, but similar ideas can be tailored.
Also for split graphs, our general tools cannot be applied,
but the NP-completeness of the few-color case follows as a corollary to some discussions for grid-like graphs.
The complexity of the colorful case on split graphs remains unsettled.

\paragraph{Related results}

\sts{} can be seen as a variant of the famous \emph{15 puzzle}.
The 15 puzzle is played on a $4 \times 4$ board with $16$ cells.
On the board, there are 15 pieces numbered from $1$ to $15$ and one vacant cell.
In each turn, we can slide an adjacent piece to the vacant cell.
The goal is to place the pieces at the right positions (see \cref{fig:15-puzz}).
By regarding the vacant cell (instead of an adjacent piece) as the piece moving in each step,
we can see the sliding process in the 15 puzzle as a swapping sequence on the $4 \times 4$ grid that starts at the vacant cell.
If we define the \emph{generalized 15 puzzle} as the same problem considered on general graphs with arbitrary colorings,
then it is almost the same as \sts{},
and the difference is whether the first and last vertices in the walk corresponding to a swapping sequence
are specified in the input (for the generalized 15 puzzle) or not (for \sts{}).

The generalized 15 puzzle has been extensively studied with respect to the ``reachability'',
i.e., under the setting where the question is the existence of a swapping sequence (not the minimum length).
It was shown by Johnson and Story~\cite{JohnsonS1879} that the reachability in the original 15 puzzle
can be decided from the parity of the total distance from the initial to final token placements.
This was later generalized further and a characterization of the reachability was given.
For example, 
it is easy to see that the characterization given by Trakultraipruk~\cite{Trakultraipruk2013ConnectivityPO} is polynomial-time testable.
On the other hand, the problem of finding a swapping sequence of the minimum length has been studied only for a couple of cases.
It was shown by Ratner and Warmuth~\cite{RatnerW90} that the generalized 15 puzzle is NP-complete on $n \times n$ grids,
in which case the problem is called the \emph{$(n^{2}-1)$ puzzle}.
A short proof for the same result was presented later by Demaine and Rudoy~\cite{DemaineR18}.

\begin{figure}[tb]
  \centering
  \includegraphics[scale=.4]{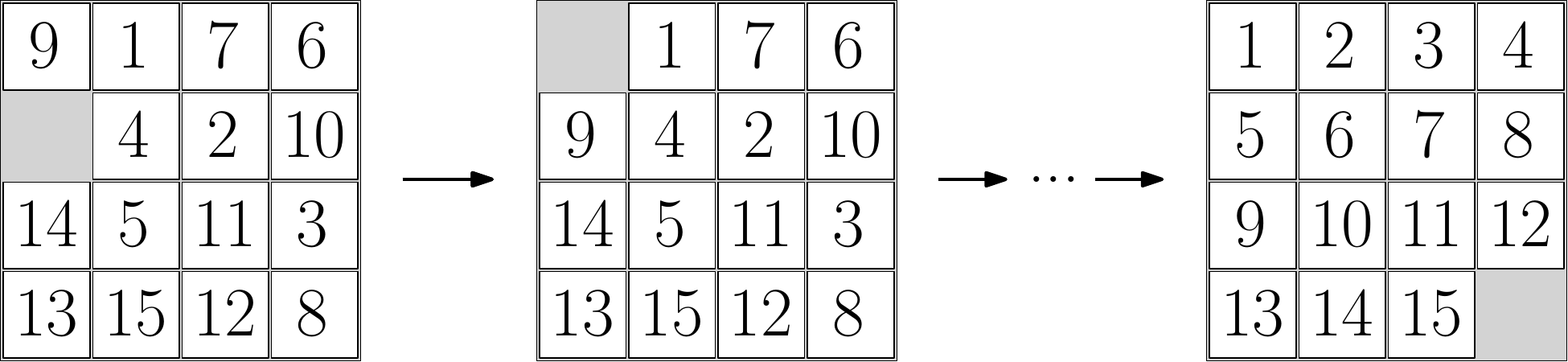}
  \caption{The 15 puzzle. Each step can be seen as a move of the vacant cell.}
  \label{fig:15-puzz}
\end{figure}

Although \sts{} is quite close to the generalized 15 puzzle,
its concept comes also from its non-sequential variant \textsc{Token Swapping},
which does not ask for the existence of a walk consistent with a swapping sequence
but allows to swap the tokens on the endpoints of any edge in each step.
The complexity of \textsc{Token Swapping} has shown to be quite different from its sequential variant.
For example, it is recently shown that \textsc{Token Swapping} is NP-complete even on trees~\cite{AichholzerDKLLMRWW22}.

The generalized 15 puzzle and \textsc{(Sequential) Token Swapping} are 
sometimes considered in \emph{combinatorial reconfiguration} as well.
See the surveys~\cite{Heuvel13,Nishimura18} for the background and related results in this context.

% preliminaries
\section{Preliminaries}
\label{sec:pre}

We use standard terminologies for graphs (see e.g.,~\cite{Diestel16_5th} for the terms not defined here).
Let $G = (V,E)$ be an undirected graph.
For $S \subseteq V$, the subgraph of $G$ induced by $S$ is denoted by $G[S]$.
A sequence $W = \langle w_{1}, \dots, w_{|W|} \rangle$ of vertices is a \emph{walk} of length $|W|-1$ in $G$
if $\{w_{i}, w_{i+1}\} \in E$ for $1 \le i < |W|$.
A vertex of a connected graph is a \emph{cut vertex} if the removal of the vertex makes the graph disconnected.
A connected graph is \emph{biconnected} if it contains no cut vertex.
A maximal induced biconnected subgraph of a graph is called a \emph{biconnected component} of the graph.
Let $\mathcal{B}_{G}$ denote the set of biconnected components of $G$.
It is known that $\mathcal{B}_{G}$ can be computed in linear time~\cite{HopcroftT73}.

A graph is a \emph{cactus} if each biconnected component is a cycle or a 2-vertex complete graph.
A graph is a \emph{block graph} if each biconnected component is a complete graph.
A graph is a \emph{block-cactus graph} if each biconnected component is a cycle or a complete graph.
A \emph{chordal graph} is a graph with no induced cycle of length $4$ or more.
A \emph{chordal bipartite graph} is a bipartite graph with no induced cycle of length $6$ or more.
A graph is a \emph{split graph} if its vertex set can be partitioned into a clique and an independent set.

The \emph{$h \times w$ grid} has the vertex set $V = \{1,\dots,h\} \times \{1,\dots,w\}$
and the edge set $\{\{(x,y), (x',y')\} \mid (x,y), (x',y') \in V, \; |x-x'| + |y-y'| = 1\}$.
A graph is a \emph{grid} if it is the $h \times w$ grid for some integers $h$ and $w$.
A graph is a \emph{grid graph} if it is an induced subgraph of some grid.
We say that a grid graph $G = (V,E)$ is given with a \emph{grid representation}
if $V \subseteq \mathbb{Z}^{2}$ and 
$E = \{\{(x,y), (x',y')\} \mid (x,y), (x',y') \in V, \; |x-x'| + |y-y'| = 1\}$.
The \emph{$h \times w$ king's graph} is obtained from the $h \times w$ grid by adding all diagonal edges of the unit squares (4-cycles) in the grid;
that is, the vertex set is $V = \{1,\dots,h\} \times \{1,\dots,h\}$
and the edges set is $\{\{(x,y), (x',y')\} \mid (x,y), (x',y') \in V, \; \max\{|x-x'|, |y-y'|\} = 1\}$.
A graph is a \emph{king's graph} if it is the $h \times w$ king's graph for some integers $h$ and $w$.
We call a vertex $(x,y)$ of a king's graph \emph{even} if $x+y$ is even
and \emph{odd} if $x+y$ is odd.
In passing, the name of a king's graph comes from the legal moves of the king chess piece on a chessboard.

As mentioned in \cref{sec:intro}, the generalized 15 puzzle
can be seen as a variant of \sts{} with the first and last vertices specified.
In the following, we call it \ststs{$s$}{$t$}\@.
\begin{description}
  \item[Problem:] \ststs{$s$}{$t$}
  \item[Input:] A graph $G = (V, E)$, colorings $f, f'$ of $G$, $s, t \in V$, and an integer $k$.
  \item[Question:] Is there a swapping sequence of length at most $k$ from $f$ to $f'$
  such that the corresponding walk starts at $s$ and ends at $t$?
\end{description}
Note that $s$ and $t$ in an instance of \ststs{$s$}{$t$} are not necessarily distinct.

% algorithms
\section{Polynomial-time algorithm for block-cactus graphs}
\label{sec:algorithm}
In this section, we present a polynomial-time algorithm for \sts{} on block-cactus graphs.
We prove the following theorem.
\begin{theorem}
\label{thm:block-cactus}
\sts{} on block-cactus graphs can be solved in $O(n^{3})$ time,
where $n$ is the number of vertices.
\end{theorem}
Note that although \cref{thm:block-cactus} is stated for \sts, which is a decision problem,
the algorithm presented below actually solves the optimization version of the problem in the same running time.
That is, it computes the minimum length of a swapping sequence from $f$ to $f'$ in $O(n^{3})$ time.

The main part of the algorithm is the subroutine for solving \ststs{$s$}{$t$}\@. 
Given that subroutine, the algorithm just tries all pairs of vertices as the first and last vertices.
In the following, we focus on this subroutine.

We show that the problem on a graph can be reduced 
to a generalized problem on its biconnected components.
Then it suffices to show that the generalized problem can be solved in polynomial time on complete graphs and cycles.
We prove this in a way similar to Yamanaka et al.~\cite{YamanakaDHKNOSS19}
but the proofs here are much more involved because of the generality of the problem.

\subsection{Reduction to a generalized problem on biconnected components}

We generalize \ststs{$s$}{$t$} by adding a subset $P$ of vertices to be visited as follows.
\begin{description}
  \item[Problem:] \substs
  \item[Input:] A graph $G = (V, E)$, colorings $f, f'$ of $G$, $s, t \in V$, and $P \subseteq V$.
  \item[Task:] Find the minimum length of a swapping sequence from $f$ to $f'$ (if any exists)
  such that the corresponding walk $W = \langle w_{1}, w_{2}, \dots, w_{|W|} \rangle$ satisfies 
  that $w_{1} = s$, $w_{|W|} = t$, and $P \subseteq \{w_{1}, w_{2}, \dots, w_{|W|}\}$.
\end{description}

Let $\lensubsts(G, f, f', s, t, P)$ denote the answer for the instance $\langle G, f, f', s, t, P \rangle$ of \substs\@.
We set it to $\infty$ if no swapping sequence from $f$ to $f'$ exists.
Note that $\lensubsts(G, f, f', s, t, \emptyset)$ is the minimum $k$
such that $\langle G, f, f', s, t, k \rangle$ is a yes-instance of \ststs{$s$}{$t$}\@.

Let $\langle G, f, f', s, t, k\rangle$ be an instance of \ststs{$s$}{$t$} and let $H$ be a biconnected component of $G$.
Let us see how a solution to this instance passes through $H$.
If $s \notin V(H)$, then the first vertex visited in $H$ is the cut vertex closest to $s$.
Similarly, if $t \notin V(H)$, then the last vertex visited in $H$ is the cut vertex closest to $t$.
Also, a cut vertex $u$ of $G$ belonging to $H$ has to be visited
if at least one vertex in $H$ is visited and there is a vertex $v \notin V(H)$ such that 
$f(v) \ne f'(v)$ and $u$ is the closest vertex in $H$ to $v$.
With these observations, 
we construct an instance $\langle H, f_H, f'_H, s_H, t_H, P_H \rangle$ of \substs{} as follows,
where $c_{v}$ is the cut vertex in $H$ that separates $v$ and $V(H)$.
\begin{itemize}
  \item Set $f_{H} = f|_{V(H)}$. If $s \notin V(H)$, then update $f_{H}$ as $f_{H}(c_{s}) \coloneqq f(s)$.
  \item Set $f'_{H} = f'|_{V(H)}$. If $t \notin V(H)$, then update $f'_{H}$ as $f'_{H}(c_{t}) \coloneqq f(s)$.
  \item Set $s_{H} = s$ if $s \in V(H)$. Otherwise, set $s_{H} = c_{s}$.
  \item Set $t_{H} = t$ if $t \in V(H)$. Otherwise, set $t_{H} = c_{t}$.
  \item Set $P_{H}$ to the set of cut vertices $c_{v}$ of $G$ belonging to $V(H)$
  such that $c_{v}$ separates $v$ and $H$ for some $v \notin V(H)$ with $f(v) \ne f'(v)$.
\end{itemize}
The following lemma says that this instance correctly captures how a solution to \ststs{$s$}{$t$} on $G$ affects $H$.

\begin{lemma}
\label{lem:sub-pd-sum}
For a graph $G$, colorings $f, f'$ of $G$, and $s,t \in V$,
\[
\lensubsts(G, f, f', s, t, \emptyset) = \sum_{H \in \mathcal{B}_{G}} \lensubsts(H, f_H, f'_H, s_H, t_H, P_H).
\]
\end{lemma}
\begin{proof}
We use induction on $|\mathcal{B}_{G}|$, the number of biconnected components of $G$.
If $|\mathcal{B}_{G}| = 1$, then the unique biconnected component is $G$ itself and thus the statement holds.
In the following, we assume that $|\mathcal{B}_{G}| = b + 1$ for some $b \ge 1$
and that the statement holds for all graphs with at most $b$ biconnected components.

Let $S \in \mathcal{B}_{G}$ be an arbitrary biconnected component that contains $s$.
Since $|\mathcal{B}_{G}| = b+1 \ge 2$, the biconnected component $S$ has at least one cut vertex.
Let $c_{1}, \dots, c_{a}$ be the cut vertices of $G$ in $H$.
Let $G_{1}, \dots, G_{a}$ be the nontrivial connected components of $G - E(S)$
such that $c_{i} \in V(G_{i})$ for each $i$ (see \cref{fig:cut-vertex}).
\begin{figure}[tb]
  \centering
  \includegraphics[scale=.7]{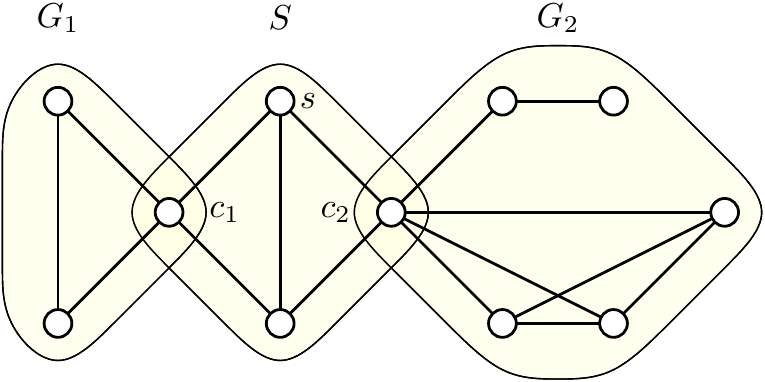}
  \caption{Cut vertices $c_{1}, c_{2}$ and the corresponding subgraphs $G_{1}, G_{2}$.}
  \label{fig:cut-vertex}
\end{figure}

For each $G_{i}$, we set $f^{(i)}$, $f^{\prime (i)}$, $s^{(i)}$, and $t^{(i)}$ as follows.
\begin{itemize}
  \item Set $f^{(i)} = f|_{V(G_{i})}$, and then update it as $f^{(i)}(c_{i}) \coloneqq f(s)$.
  \item Set $f^{\prime (i)} = f'|_{V(G_{i})}$. If $t \notin V(G_{i})$, then update it as $f^{\prime (i)}(c_{i}) \coloneqq f(s)$.
  \item Set $s^{(i)} = c_{i}$.
  \item Set $t^{(i)} = t$ if $t \in V(G_{i})$. Otherwise, set $t^{(i)} = c_{i}$.
\end{itemize}
For each biconnected component $H$ of $G_{i}$, we define $f^{(i)}_{H}$, $f^{\prime (i)}_{H}$, $s^{(i)}_{H}$, $t^{(i)}_{H}$, $P^{(i)}_{H}$ 
in the same way as $f_{H}, f'_{H}, s_{H}, t_{H}, P_{H}$.
Observe that each biconnected component of $G_{i}$ is a biconnected component of $G$ as well.
This implies that $|\mathcal{B}_{G_{i}}| \le b$ for each $i$ (as $S$ is missing).
By the induction hypothesis, the statement of the lemma holds for each $G_{i}$. 
To be more precise, it holds for each $i$ that
\begin{align}
  \lensubsts(G_{i}, f^{(i)}, f^{\prime (i)}, s^{(i)}, t^{(i)}, \emptyset)
  =
  \sum_{H \in \mathcal{B}_{G_{i}}} \lensubsts(H, f^{(i)}_{H}, f^{\prime (i)}_{H}, s^{(i)}_{H}, t^{(i)}_{H}, P^{(i)}_{H}).
  \label{eq:substs-ind-hyp}
\end{align}
We can see that for every $G_{i}$ and every $H \in \mathcal{B}_{G_{i}}$, it holds that
$f^{(i)}_{H} = f_{H}$, $f^{\prime (i)}_{H} = f'_{H}$, $s^{(i)}_{H} = s_{H}$, $t^{(i)}_{H} = t_{H}$,
and $P^{(i)}_{H} \cup \{s^{(i)}_{H}\} = P_{H}$. This implies that
\[
  \lensubsts(H, f_{H}, f'_{H}, s_{H}, t_{H}, P_{H}) = \lensubsts(H, f^{(i)}_{H}, f^{\prime (i)}_{H}, s^{(i)}_{H}, t^{(i)}_{H}, P^{(i)}_{H}),
\]
and thus by \cref{{eq:substs-ind-hyp}}, 
\[
  \lensubsts(G_{i}, f^{(i)}, f^{\prime (i)}, s^{(i)}, t^{(i)}, \emptyset)
  =
  \sum_{H \in \mathcal{B}_{G_{i}}} 
  \lensubsts(H, f_{H}, f'_{H}, s_{H}, t_{H}, P_{H}).
\]
Now, since $\mathcal{B}_{G} = \{S\} \cup \bigcup_{1 \le i \le a} \mathcal{B}_{G_{i}}$,
\begin{multline}
  \sum_{H \in \mathcal{B}_{G}} \lensubsts(H, f_H, f'_H, s_H, t_H, P_H)
  = %\\[-3ex]
  \lensubsts (S, f_{S}, f'_{S}, s_{S}, t_{S}, P_{S}) 
  +
  \sum_{i=1}^{a}\lensubsts(G_{i}, f^{(i)}, f^{\prime (i)}, s^{(i)}, t^{(i)}, \emptyset).
  \label{eq:substs}
\end{multline}

\begin{claim}
  \label{clm:lensubsts-ge}
  The following inequality holds:
  \[
    \lensubsts(G, f, f', s, t, \emptyset) \ge
    \lensubsts (S, f_{S}, f'_{S}, s_{S}, t_{S}, P_{S}) +
    \sum_{i=1}^{a} \lensubsts(G_{i}, f^{(i)}, f^{\prime (i)}, s^{(i)}, t^{(i)}, \emptyset).
  \]
\end{claim}
\begin{proof}[\cref{clm:lensubsts-ge}]
We assume that $\lensubsts(G, f, f', s, t, \emptyset) \ne \infty$ since otherwise the claim is trivially true.
Let $W = \langle w_{1}, \dots, w_{\lensubsts(G, f, f', s, t, \emptyset)+1} \rangle$
be a walk that corresponds to a swapping sequence from $f$ to $f'$ such that $w_{1} = s$ and $w_{|W|} = t$.

\paragraph{Constructing a walk in $G^{(i)}$.}
If $f^{(i)} = f^{\prime (i)}$ and $s^{(i)} = t^{(i)}$, then the trivial walk $\langle s^{(i)} \rangle$ certificates
that $\lensubsts(G_{i}, f^{(i)}, f^{\prime (i)}, s^{(i)}, t^{(i)}, \emptyset) = 0$.
Otherwise, we construct a shortest walk for each $G_{i}$ using the walk $W$.

Let $W_{1}, \dots, W_{l}$ be all maximal subwalks of $W$ appearing in this ordering
such that each $W_{j}$ contains vertices of $V(G_{i})$ only and $|W_{j}| \ge 2$.
Note that there is at least one such maximal subwalk
since the moving token needs to swap tokens in $G_{i}$ as $f^{(i)} \ne f^{\prime (i)}$ or $s^{(i)} \ne t^{(i)}$ holds.
By the definition of $G_{i}$, the cut vertex $c_{i}$ ($= s^{(i)}$)
is the unique vertex in $V(G_{i})$ that is adjacent to a vertex not in $V(G_j)$.
This implies that, for each $p < l$, $W_{p}$ is a walk from $c_{i}$ to $c_{i}$, and $W_{l}$ is a walk from $c_{i}$ to $t^{(i)}$.
Thus we can concatenate the walk $W_{1}, \dots, W_{l}$ into one walk $W^{(i)}$ from $s^{(i)}$ to $t^{(i)}$ such that
$W^{(i)} = \langle w_{1, 1}, \dots, w_{1, |W_{1}|} \, (= w_{2, 1}), \, w_{2, 2}, \dots,
  w_{2, |W_{2}|} \, (=w_{3,1}), \,  \dots, w_{l-1, |W_{l-1}|} \, (= w_{l, 1}), \dots w_{l, |W_{l}|} \rangle$,
where $w_{p, q}$ is the $q$th vertex in $W_{p}$.
Now we show that $W^{(i)}$ corresponds to a desired swapping sequence on $G_{i}$.

Observe that the moving token (walking along $W$) brings into $G_{i}$
a token $t$ not originally in $G_{i}$ when it leaves $G_{i}$ from $c_{i}$ right after visiting $w_{p, |W_{p}|}$ for some $p$.
After that, when the moving token visits $w_{p+1, |W_{p}|}$, the token $t$ leaves $G_{i}$ and 
the token placement restricted to $G_{i}$ becomes exactly the same as the one right when $w_{p, |W_{p}|}$ had been visited.
Thus, by swapping along the walk $W^{(i)}$ we can replicate the swappings in $G^{(i)}$ by the swapping sequence corresponding to $W$.
Recall that $f^{(i)}$ is the same as $f$ on $V(G_{i})$ with an exception $f^{(i)}(s^{(i)}) = f(s)$.
Since $\lensubsts(G, f, f', s, t, \emptyset) \ne \infty$, we have $f(s) = f'(t)$.
Hence, if $t \in V(G_{i})$, then $f^{\prime (i)}(t^{(i)}) = f^{\prime (i)}(t) = f'(t) = f(s)$.
Otherwise, $f^{\prime (i)}(t^{(i)}) = f^{\prime (i)}(c_{i}) = f(s)$.
Thus, $W^{(i)}$ corresponds to a desired swapping sequence of length $\sum_{p=1}^{l}(|W_{p}| - 1)$
for the instance $\langle G_{i}, f^{(i)}, f^{\prime (i)}, s^{(i)}, t^{(i)}, \emptyset \rangle$ of \substs\@.

\paragraph{Constructing a walk in $S$.}

As before, we construct a walk $W^{(S)}$ by concatenating the maximal subwalks of $W$ passing through $S$.
In almost the same way as before, we can see that the swapping sequence along $W^{(S)}$ is the desired one.
The only difference is the additional requirement for visiting the vertices in $P_{S}$.
A cut vertex $c_{i}$ belongs to $P_{S}$ if there exists a vertex $v \in V(G_{i}) \subseteq V(S)$ such that $f(v) \ne f(v')$.
The walk $W$ has to visit $c_{i}$ since it starts at $s \in V(S)$, and thus $W^{(S)}$ visits $c_{i}$.

\paragraph{The total length of the walks.}
From the discussions above, we have
\begin{multline*}
  (|W^{(S)}|-1) + \sum_{i=1}^{a} (|W^{(i)}|-1)
  %\\[-2.5ex]
  \ge
  \lensubsts (S, f_{S}, f'_{S}, s_{S}, t_{S}, P_{S}) +
  \sum_{i=1}^{a} \lensubsts(G_{i}, f^{(i)}, f^{\prime (i)}, s^{(i)}, t^{(i)}, \emptyset).
\end{multline*}
Observe that each consecutive vertices $w_{i}, w_{i+1}$ in the walk $W$
contributes $1$ to exactly one of the lengths $|W^{(S)}|-1, |W^{(1)}|-1, \dots, |W^{(a)}|-1$.
Thus, $|W|-1 =  (|W^{(S)}|-1) + \sum_{i=1}^{a} (|W^{(i)}|-1)$.
Since $|W|-1 = \lensubsts(G, f, f', s, t, \emptyset)$, the claim follows.
(\textit{The end of the proof of \cref{clm:lensubsts-ge}.})
\end{proof}

\begin{claim}
  \label{clm:lensubsts-le}
  The following inequality holds:
  \[
    \lensubsts(G, f, f', s, t, \emptyset) \le
    \lensubsts (S, f_{S}, f'_{S}, s_{S}, t_{S}, P_{S}) +
    \sum_{i=1}^{a} \lensubsts(G_{i}, f^{(i)}, f^{\prime (i)}, s^{(i)}, t^{(i)}, \emptyset).
  \]
\end{claim}
\begin{proof}[\cref{clm:lensubsts-le}]
We assume that the right-hand side is not $\infty$ since otherwise the claim clearly holds.
Let $W^{(S)} = \langle w_{S, 1}, w_{S, 2}, \dots, w_{S, |W^{(S)}|} \rangle$
be a walk corresponding to a desired swapping sequence of length $\lensubsts (S, f_{S}, f'_{S}, s_{S}, t_{S}, P_{S})$.
Similarly, for each $i$, let $W^{(i)} = \langle w_{i, 1}, w_{i, 2}, \dots, w_{i, |W^{(i)}|} \rangle$
be a walk corresponding to a desired swapping sequence of length $\lensubsts(G_{i}, f^{(i)}, f^{\prime (i)}, s^{(i)}, t^{(i)}, \emptyset)$.

We first assume that $t \in V(S)$.
In this case, $s^{(i)} = t^{(i)} = c_{i}$ holds for each $G_{i}$.
If $f^{(i)} \ne f^{\prime (i)}$,  then $c_{i} \in P_{S}$ holds,
and thus there exists $l$ such that $w_{S, l} = c_{i}$.
Replacing the vertex $w_{S, l}$ in $W^{(S)}$ with the walk $W^{(i)}$,
we obtain a new walk from $s_{S}$ to $t_{S}$ as follows:
\[
 \langle w_{S, 1} \dots, w_{S, l-1}, c_{i}, w_{i, 2}, \dots, w_{i, |W^{(i)}|-1}, c_{i}, w_{S, l+1}, \dots w_{S, |W^{(S)}| }\rangle. 
\]
Since $W^{(S)}$ and $W^{(i)}$ share $c_{i}$ only, 
after applying the swapping sequence along the new walk, each vertex $v$ in $S$ has color $f'_{S}(v) = f'(v)$
and each vertex $v \in V(G_{i}) \setminus \{c_{i}\}$ has color $f^{\prime (i)}_{S}(v) = f'(v)$.
We repeat the replacement for all $G_{i}$. We call the obtained walk $W$.
The walk $W$ is a walk from $s$ ($=s_{S}$) to $t$ ($=t_{S}$) that changes the coloring from $f$ to $f'$.
Since each replacement for $G_{i}$ increases the length of $W$ by $|W^{(i)}| - 1$,
the length of $W$ can be bounded as follows:
\begin{align*}
|W| - 1
&\le |W^{(S)}| - 1 + \sum_{i=1}^{a} (|W^{(i)} - 1) \\[-2ex]
&= 
\lensubsts(S, f_{S}, f'_{S}, s_{S}, t_{S}, P_{S}) 
+ \sum_{i=1}^{a} \lensubsts(G_{i}, f^{(j)}, f^{\prime (i)}, s^{(i)}, t^{(i)}, \emptyset).
\end{align*}

We next assume that $t \notin V(S)$.
Let $j$ be the unique index such that $t \in V(G_{j})$.
Now we have, $s_{S} = s$, $t_{S} = c_{j}$, $s^{(j)} = c_{j}$, $t^{(j)} = t$,
and $s^{(i)} = t^{(i)} = c_{i}$ for each $i \ne j$.
By the same discussion as the previous case, 
we can construct a walk $W$ from $s$ to $c_{j}$ such that
the corresponding swapping sequence changes the colors of each vertex $v \notin V(G_{j})$ to $f'(v)$,
while the vertices in $V(G_{j}) \setminus \{c_{j}\}$ are left untouched.
Since the vertex of $G_{j}$ does not appear in $W$ except for $c_{j}$,
the coloring obtained by $W$ restricted to $G_{j}$ coincides with $f^{(j)}$.
Therefore, we can obtain $f'$ by attaching $W^{(j)}$ to the end of $W$
by identifying the last vertex of $W$ and the first vertex of $W^{(j)}$.
The length of the resultant walk is $|W^{(S)}| - 1 + \sum_{i=1}^{a} (|W^{(i)} - 1)$ as before.
(\textit{The end of the proof of \cref{clm:lensubsts-le}.})
\end{proof}

Now the lemma follows by \cref{eq:substs,clm:lensubsts-ge,clm:lensubsts-le}.
\end{proof}

\subsection{\substs{} on cycles}

\begin{figure}[tb]
  \centering
  \includegraphics[width=\textwidth]{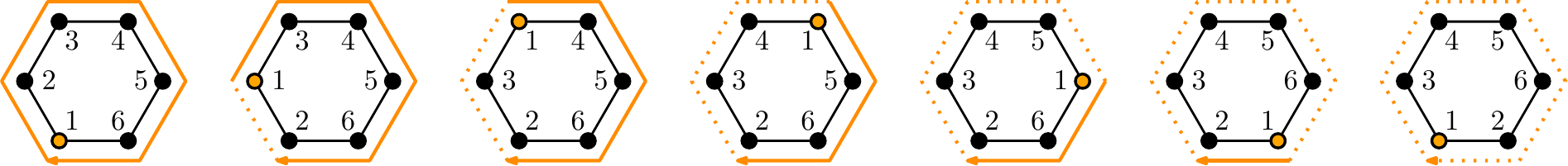}
  \caption{A swapping sequence on a cycle.}
  \label{fig:cycle_eg}
\end{figure}

\begin{lemma}
\label{lem:subpd-cycle}
\substs{} on cycles can be solved in linear time.
\end{lemma}
\begin{proof}
Let $\langle C, f, f', s, t, P \rangle$ be an instance of \substs{}, where $C$ is a cycle of $n$ vertices.
We assume that $f(s) = f'(t)$ since otherwise it is a trivial no-instance.
We arbitrarily fix a cyclic orientation on $C$ and call it the clockwise direction
(and the other one the counterclockwise direction).

Observe that if the moving token goes in one direction on the cycle, 
then the other tokens passed are shifted to the other direction (see \cref{fig:cycle_eg}).
Observe also that if the moving token goes one step in one direction
and goes back in the other direction immediately, then these moves cancel out and the coloring stays the same.
Thus, if $P = \emptyset$, then an optimal solution never goes back and forth.
Based on these observations, Yamanaka et al.~\cite{YamanakaDHKNOSS19} presented
a polynomial-time algorithm for \sts{} on cycles.
We also use these facts, but since $P \ne \emptyset$ in general, we need some new ideas.

Let $W = \langle u_{1}, \dots, u_{p} \, (= v_{1}), \dots, v_{q} \, (= w_{1}), \dots, w_{r}\rangle$
be a walk corresponding to a desired swapping sequence of the minimum length, where
\begin{itemize}
  \item $v_{1}$ is the last vertex in $W$ such that $v_{1} = s$ and 
  the coloring after executing the swapping sequence up to $v_{1}$ is $f$, and
  \item $v_{q}$ is the first vertex in $W$ such that $v_{q} = t$ and 
  the coloring after executing the swapping sequence up to $v_{q}$ is $f'$.
\end{itemize}
We show that there is a direction $\leftarrow$, which is clockwise or counterclockwise,
such that the following properties hold:
\begin{itemize}
  \item the moves along $\langle u_{1}, \dots, u_{p} \rangle$
  first go in the direction $\leftarrow$ some number of steps and 
  then go back in the opposite direction $\rightarrow$ the same number of steps;
  \item the moves along $\langle v_{1}, \dots, v_{q} \rangle$ go in the direction $\rightarrow$ only;
  \item the moves along $\langle w_{1}, \dots, w_{r} \rangle$ 
  first go in the direction $\rightarrow$ some number of steps and 
  then go back in the direction $\leftarrow$ the same number of steps.
\end{itemize}

To show the property of $\langle v_{1}, \dots, v_{q} \rangle$,
assume that $q \ge 2$ and that $v_{2}$ is the clockwise neighbor of $v_{1}$.
If $v_{i} = s$ for some $i \ne 1$, then $V(C) = \{v_{1}, \dots, v_{q}\}$ holds
as the coloring after executing the swapping sequence up to $v_{i}$ is not $f$.
Similarly, if $v_{i} = t$ for some $i \ne q$, then $V(C) = \{v_{1}, \dots, v_{q}\}$ holds
as the coloring after executing the swapping sequence up to $v_{i}$ is not $f'$.
Otherwise, $\{v_{1}, \dots, v_{q}\}$ is the set of consecutive vertices on $C$
from $s$ to $t$ in the clockwise direction.
In all cases, if there is a counterclockwise move,
then the first such move $v_{j} \rightarrow v_{j+1}$ can be removed with the previous one $v_{j-1} \rightarrow v_{j}$
without changing the set of visited vertices.
Since $W$ is of the minimum length, we can conclude that there is no such move.
Now the properties of $\langle u_{1}, \dots, u_{p} \rangle$ and $\langle w_{1}, \dots, w_{r} \rangle$
follows easily as they are necessary only for visiting more vertices (in $P$).

We compute the minimum length for each of the cases 
$V(C) \ne \{v_{1}, \dots, v_{q}\}$ and $V(C) = \{v_{1}, \dots, v_{q}\}$.

\paragraph{The case of $V(C) \ne \{v_{1}, \dots, v_{q}\}$.}
In this case, $\langle v_{1}, \dots, v_{q} \rangle$ is a simple path from $s$ to $t$.
We assume that this is a clockwise path. (The other case is symmetric.)
Since only $\langle v_{1}, \dots, v_{q} \rangle$ changes the coloring and the other parts of $W$ cancel out,
we first check that we get $f'$ by applying $\langle v_{1}, \dots, v_{q} \rangle$ to $f$
and then compute the other parts that visit $P \setminus \{v_{1}, \dots, v_{q}\}$.
Let $\langle x_{1}, \dots, x_{k} \rangle$ be the sequence of the vertices of $P \setminus \{v_{1}, \dots, v_{q}\}$
ordered in the counterclockwise order from $s$ to $t$.
Observe that if $p \ge 2$, then the first part $\langle u_{1}, \dots, u_{p}\rangle$ of $W$
starts at $s$ in the counterclockwise direction, visits some vertices $x_{1},\dots,x_{k'}$,
and comes back to $s$. Thus, its length is the twice of the distance from $s$ to $x_{k'}$ in the counterclockwise direction.
Similarly, if $r \ge 2$, then the last part $\langle w_{1}, \dots, w_{r}\rangle$ of $W$
starts at $t$ in the clockwise direction, visits the remaining vertices $x_{k'+1}, \dots, x_{k}$,
and comes back to $t$. Its length is the twice of the distance from $t$ to $x_{k'+1}$ in the clockwise direction.
The index $k' \in \{0, \dots, k\}$ that minimizes the sum $p+r$ can be found in linear time
by precomputing the counterclockwise distances from $s$ to $x_{1},\dots,x_{k}$
and the clockwise distances from $t$ to $x_{1},\dots,x_{k}$ in linear time.

\paragraph{The case of $V(C) = \{v_{1}, \dots, v_{q}\}$.}
In this case, $W = \langle v_{1}, \dots, v_{q} \rangle$ as there is no other vertex to visit.
Now it is easy to compute the minimum length in polynomial time:
guess the direction of the walk;
go in the guessed direction $n-1$ steps from $s$;
and then further proceed in the same direction until we get the desired coloring.
Since the minimum length is $O(n^{3})$ (if not $\infty$) in general~\cite{YamanakaDHKNOSS19}, 
this algorithm runs in polynomial time.

To do it in linear time, we reduce the problem to a substring matching problem
that can be solved in linear time by the KMP algorithm~\cite{KnuthMP77}.

Assume that $W$ goes in the clockwise direction. (The other case is symmetric.)
Let $g$ be the coloring obtained from $f$ by executing the swapping sequence along $W$
up to the first point where all vertices are visited and the moving token is placed at $t$.
The remaining of the walk we are looking for repeats the (clockwise) walk from $t$ to $t$ some number of times.
Observe that if we repeat it $i$ times, then the coloring we get is the one obtained from $g$
by shifting the non-moving tokens $i$ steps in the counterclockwise direction~(see \cref{fig:cycle_eg}).
Thus it suffices to compute the minimum number of shifts to obtain $f'$.

Let $t_{\mathrm{next}}$ and $t_{\mathrm{prev}}$ be the clockwise and counterclockwise neighbors of $t$, respectively.
Let $S_{g} = \langle c_{1}, \dots, c_{n-1} \rangle$ be the sequence of the colors under $g$ of vertices
from $t_{\mathrm{next}}$ to $t_{\mathrm{prev}}$ in the clockwise ordering.
Similarly, let $S_{f'}$ be the same sequence but under $f'$.
Observe that if $S_{f'}$ can be obtained from $S_{g}$ by $i$ cyclic shifts 
(in the counterclockwise direction, or to the left in this context),
then $S_{f'} = \langle c_{i+1}, \dots, c_{n-1}, c_{1}, \dots, c_{i} \rangle$ holds.
The minimum $i$ satisfying this can be found by finding the first index such that $S_{f'}$ starts in
$S_{g} \cdot S_{g} =  \langle c_{1}, \dots, c_{n-1}, c_{1}, \dots, c_{n-1} \rangle$ as a substring,
which can be done in linear time~\cite{KnuthMP77}.
\end{proof}

\subsection{\substs{} on complete graphs}
Let $I= \langle K, f, f', s, t, P \rangle$ be an instance of \substs, where $K = (V, E)$ is a complete graph.
As before, we assume that $f(s) = f'(t)$.
Furthermore, we assume that $s$ has a unique color under $f$ (and so does $t$ under $f'$).
We set the unique color to $0$. That is, we assume that $f(s) = f'(t) = 0$,
$f(v) \ne 0$ if $v \ne s$, and $f'(v) \ne 0$ if $v \ne t$.
Observe that this does not change the instance since the moving token anyway moves from $s$ to $t$.

Let $R = \{v \in V \mid f(v) \ne f'(v)\} \cup \{s, t\} \cup P$.
We define a directed multigraph $D = (V_D, E_D)$, possibly with self-loops and parallel edges, as %follows:
$V_{D} = \{f(v) \mid v \in R\}$ and $E_{D} = \{(f(v), f'(v)) \mid v \in R\}$.
This graph $D$ is almost the same as the conflict graph defined in \cite{YamanakaDHKNOSS19}.
The difference here is the self-loops corresponding to the vertices in $P \setminus \{v \in V \mid f(v) \ne f'(v)\}$
(and $s$ when $s = t$).
Thus the assumption that $f$ and $f'$ use the same number of vertices for each color 
implies that the indegree and the outdegree are the same for each vertex (or, color) in $D$.
This implies that each connected component of $D$ is strongly connected and has an Eulerian circuit.
Let $\mathrm{cc}(D)$ denote the number of (strongly) connected components of $D$.

The rest of this subsection is devoted to a proof of the following equation:
\begin{equation}
  \lensubsts(I) = |R| + \mathrm{cc}(D) - 2.
  \label{eq:subpd-complete-graph}
\end{equation}

\begin{lemma}
\label{lem:subpd-complete-upper-bound}
$\lensubsts(I) \le |R| + \mathrm{cc}(D) - 2$.
\end{lemma}
\begin{proof}
We prove this by constructing a swapping sequence following the ideas in~\cite{YamanakaDHKNOSS19}.
Let $C_{1}, \dots, C_{\mathrm{cc}(D)}$ be the connected components of $D$.
We assume without loss of generality that $0 \in V(C_{1})$.

For $C_{1}$, let $\langle e_{1}(1), \dots, e_{1}(t_{1}) \rangle$ be an Eulerian circuit of $C_{i}$
such that $e_{1}(1) = (f(s), f'(s))$ and $e_{1}(t_{1}) = (f(t), f'(t))$.
Such an Eulerian circuit exists since $0$ has the unique outneighbor $f'(s)$ and the unique inneighbor $f(t)$.
For each $C_{i}$ with $i \ge 2$,
let $\langle e_{i}(1), \dots, e_{i}(t_{i}) \rangle$ be an arbitrary Eulerian circuit of $C_{i}$.

Let us fix a bijective correspondence between $R$ and $E_{D}$
such that if $v \in R$ corresponds to $e \in E_{D}$, then $e = (f(v), f'(v))$ holds.
Then, for each $i$ and $j$, let $v_{i}(j) \in R$ be the vertex corresponding to $e_{i}(j)$.
Now we define a walk $W$ from the Eulerian circuits above as follows:
\begin{align*}
W = \langle \, v_{1}(1), \dots, v_{1}(t_{1}), \quad 
&v_{2}(1), \dots, v_{2}(t_{2}), \quad v_{1}(t_{1}), \\
&v_{3}(1), \dots, v_{3}(t_{3}), \quad v_{1}(t_{1}), \\
&\ldots \\
&v_{\mathrm{cc}(D)}(1), \dots, v_{\mathrm{cc}(D)}(t_{\mathrm{cc}(D)}), \quad v_{1}(t_{1}) \, \rangle.
\end{align*}

Clearly, $W$ is a walk from $s$ ($=v_{1}(1)$) to $t$ ($=v_{1}(t_{1})$).
It is easy to see that each vertex in $P$ appears in $W$.
Observe that $|W|-1 = |R| + \mathrm{cc}(D) - 2$
since $v_{1}(t_1)$ appears $\mathrm{cc}(D)$ times in $W$
each vertex in $R \setminus \{v_{1}(t_1)\}$ appears once,
and no other vertex appears in $W$.

Let $f_{W}$ be the coloring obtained by applying the swapping sequence along $W$ to $f$.
It suffices to show that $f_{W}(v) = f'(v)$ for each $v \in R$ 
since vertices in $V \setminus R$ agree in $f$ and $f'$ and do not appear in $W$.

First assume that $v = v_{1}(t_1)$.
Since the last vertex of $W$ is $v_{1}(t_1)$, it holds that $f_{W}(v_{1}(t_1)) = f(v_{1}(1)) = f(s) = 0 = f'(t)$.

Next assume that $v = v_{i}(j)$ with $j < t_{i}$.
In this case, $f_{W}(v_{i}(j)) = f(v_{i}(j + 1))$ holds
since the moving token visits $v_{i}(j)$ only once
and it visits $v_{i}(j+1)$ for the first time right after visiting $v_{i}(j)$.
Since the edges $e_{i}(j) = (f(v_{i}(j))$, $f'(v_{i}(j)))$ and $e_{i}(j+1) = (f(v_{i}(j+1))$, $f'(v_{i}(j+1)))$ 
consecutively appear in the Eulerian circuit $C_{i}$,
it holds that $f'(v_{i}(j)) = f(v_{i}(j + 1))$,
and thus $f_{W}(v_{i}(j)) = f'(v_{i}(j))$.

Finally assume that $v = v_{i}(t_{i})$ for some $i \ne 1$.
The color $f_{W}(v_{i}(t_i))$ is the color of $v_{1}(t_1)$ when the moving token visits $v_{i}(t_i)$.
The vertex $v_{1}(t_1)$ got this color $f_{W}(v_{i}(t_i))$ when the moving token visits $v_{i}(1)$.
Since $v_{i}(1)$ is visited only once, $f_{W}(v_{i}(t_i)) = f(v_{i}(1))$ holds.
Since the edges $e_{i}(t_{i}) = (f(v_{i}(t_{i}))$, $f'(v_{i}(t_{i})))$ and $e_{i}(1) = (f(v_{i}(1))$, $f'(v_{i}(1)))$ 
consecutively appear in $C_{i}$, it holds that $f'(v_{i}(t_{i})) = f(v_{i}(1))$,
and thus $f_{W}(v_{i}(t_{i})) = f'(v_{i}(t_{i}))$.
\end{proof}

\begin{lemma}
\label{lem:subpd-complete-lower-bound}
$\lensubsts(I) \ge |R| + \mathrm{cc}(D) - 2$.
\end{lemma}
\begin{proof}
We assume without loss of generality $s,t \notin P$.
We use induction on $\lensubsts(I)$.
If $\lensubsts(I) = 0$, then the instance $I= \langle K$, $f$, $f'$, $s$, $t$, $P \rangle$ satisfies that 
$f = f'$, $s = t$, $P = \emptyset$, and thus $R = \{s,t\} = \{s\}$.
As $f(s) = 0$ ($= f'(s)$), we have $V_{D} = \{0\}$ and $E_{D} = \{(0, 0)\}$.
Hence, $|R| + \mathrm{cc}(D) - 2 = 1 + 1 - 2 = 0 = \lensubsts(I)$.
In the following, we assume that $\lensubsts(I) = k$ for some $k \ge 1$
and that the lemma holds for every instance $I'$ with $\lensubsts(I) < k$.

Let $W = \langle w_{1}, w_{2}, \dots, w_{k+1} \rangle$ be the walk corresponding to a desired swapping sequence of length $k$.
Let $f_{2}$ be the second coloring in the swapping sequence.
That is, 
$f_{2}(w_{1}) = f(w_{2})$,
$f_{2}(w_{2}) = f(w_{1})$,
and $f_{2}(v) = f(v)$ for all $v \notin \{w_{1}, w_{2}\}$.
Let $W_{2} = \langle w_{2}, w_{3}, \dots, w_{k+1}\rangle$.
Then, $W_{2}$ is the walk corresponding to a swapping sequence for
$I_{2} = \langle K, f_{2}, f', w_{2}, t, P \rangle$.
In the same way as $R$ and $D$ for $I$, we define $R_{2}$ and $D_{2} = (V_{D_{2}}, E_{D_{2}})$ for $I_{2}$.
Observe that $\lensubsts(I_2) \le |W_{2}| -1 = k-1$.
By the induction hypothesis, $\lensubsts(I_{2}) \ge |R_{2}| + \mathrm{cc}(D_{2}) - 2$ holds.
Thus we have $k \ge |R_{2}| + \mathrm{cc}(D_{2}) - 1$.
Now, it suffices to show that 
\[
  |R_{2}| + \mathrm{cc}(D_{2}) \geq |R| + \mathrm{cc}(D) - 1.
\]
In the following, we consider the cases of $w_{2} \in R$ and $w_{2} \notin R$ separately.

\paragraph{The case of $w_{2} \in R$.}
Recall that
\begin{align*}
  R &= \{v \in V \mid f(v) \ne f'(v)\} \cup \{w_{1}, t\} \cup P, \\
  R_{2} &= \{v \in V \mid f_{2}(v) \ne f'(v)\} \cup \{w_{2}, t\} \cup P.
\end{align*}
By the assumption $w_{2} \in R$, it holds that $w_{2} \in R \cap R_{2}$.
For $v \notin \{w_{1}, w_{2}\}$, we have
$v \in R$ if and only if $v \in R_{2}$ since $f_{2}(v) = f(v)$.
Hence, $R_{2} \subseteq R$ and $R \setminus R_{2} \subseteq \{w_{1}\}$.
That is, $R = R_{2}$ or $R = R_{2} \mathbin{\dot{\cup}} \{w_{1}\}$, 
where $\dot{\cup}$ denotes the disjoint union.

\textit{$\bullet$ Subcase} $R = R_{2}$.
In this case, both $w_{1}$ and $w_{2}$ belong to $R = R_{2}$, and thus $V_{D_{2}} = V_{D}$.
Observe that $w_{1} \in R_{2}$ implies that $f_{2}(w_{1}) \ne f'(w_{1})$ or $w_{1} = t$.
The latter case also implies that $f_{2}(w_{1}) \ne f'(w_{1})$ as $f'(w_{1}) = f'(t) = 0 \ne f(w_{2}) = f_{2}(w_{1})$.
This implies that the set $E_{D_{2}}$ of edges is obtained from $E_{D}$ by removing 
$(f(w_{1}), f'(w_{1}))$ and $(f(w_{2}), f'(w_{2}))$
and adding $(f_{2}(w_{2}), f'(w_{2})) = (f(w_{1}), f'(w_{2}))$ 
and $(f_{2}(w_{1}), f'(w_{1})) = (f(w_{2}), f'(w_{1}))$.
If the colors $f(w_{1})$ and $f(w_{2})$ belong to different connected components $C_{a}$ and $C_{b}$ in $D$,
then $C_{a}$ and $C_{b}$ are merged into a single connected component of $D_{2}$ by the replacement of the edges,
while the other connected components are unaffected.
Thus, $\mathrm{cc}(D_{2}) = \mathrm{cc}(D) - 1$ holds.
On the other hand,
If the colors $f(w_{1})$ and $f(w_{2})$ belong to the same connected component $C_{a}$ in $D$,
then all colors involved in the replacement belong to $C_{a}$.
Since other connected components are unaffected,
$\mathrm{cc}(D_{2}) \ge \mathrm{cc}(D)$ holds.
Now, we can conclude that
\[
 |R_{2}| + \mathrm{cc}(D_{2}) = |R| + \mathrm{cc}(D_{2}) \ge |R| + \mathrm{cc}(D) - 1.
\]

\textit{$\bullet$ Subcase} $R= R_{2} \mathbin{\dot{\cup}} \{w_{1}\}$.
The set of edges $E_{D_2}$ is obtained from $E_D$ by removing $(f(w_{1}), f'(w_{1})))$ and $(f(w_{2}), f'(w_{2}))$ 
and adding $(f_{2}(w_{2}), f'(w_{2})) = (f(w_{1}), f'(w_{2}))$.
Since $w_{1} \notin R_{2}$, it holds that $ f'(w_{1}) = f_{2}(w_{1}) = f(w_{2})$.
Thus, $(f(w_{1}), f(w_{2})) = (f(w_{1}), f'(w_{1})) \in E_{D}$.
Since $(f(w_{1}), f'(w_{1})), (f(w_{2}), f'(w_{2})) \in E_{D}$,
the colors $f(w_{1})$, $f'(w_{1})$, $f(w_{2})$ and $f'(w_{2})$ belong to the same connected component of $D$.
This implies that $\mathrm{cc}(D_{2}) \ge \mathrm{cc}(D)$ as before, and thus
\[
  |R_{2}| + \mathrm{cc}(D_{2}) 
  %  = (|R|-1) + \mathrm{cc}(D_{2})
  \ge (|R|-1) + \mathrm{cc}(D)
  = |R| + \mathrm{cc}(D) - 1.
\]

\paragraph{The case of $w_{2} \notin R$.}
We first show that $f'(w_{1}) \ne f(w_{2})$.
Suppose to the contrary that $f(w_{2}) = f'(w_{1})$.
The assumption $w_{2} \notin R$
implies that $f'(w_{2}) = f(w_{2})$, and thus $f'(w_{2}) = f'(w_{1})$.
Since the color $f(w_{1}) = 0$ is unique in $f$,
we have $f(w_{2}) \ne 0$. Hence we have
\[
  f(w_{1}) \ne f'(w_{1}) = f'(w_{2}) = f(w_{2}).
\]
From the assumptions, $w_{1}, w_{2} \notin P$ holds.
We can see that $w_{k+1} \notin \{w_{1}, w_{2}\}$ as follows:
if $t = w_{1}$ ($= s$),
then $f(w_{1}) = f(s) = 0 = f'(t) = f'(w_{1})$ contradicting $f(w_{1}) \ne f'(w_{1})$;
if $t = w_{2}$,
then $f(w_{2}) = f'(w_{1}) \ne f(w_{1}) = 0 = f'(t) = f'(w_{2})$ contradicting $w_{2} \notin R$.
Since $K$ is a complete graph, $w_{1}$ and $w_{2}$ have the same neighborhood.
Therefore, the instance $I_{2} = \langle K, f_{2}, f', w_{2}, t, P \rangle$
can be seen as the one obtained from $I = \langle K, f, f', w_{1}, t, P \rangle$
by renaming $w_{1}$ as $w_{2}$ and $w_{2}$ as $w_{1}$.
This implies that $\lensubsts(I) = \lensubsts(I_{2})$, a contradiction.
Thus, $f'(w_{1}) \ne f(w_{2})$ holds.

By the assumption, $w_{2} \notin R$.
Since $f_{2}(w_{1}) = f(w_{2})$ and $f'(w_{1}) \ne f(w_{2})$,
we have $f_{2}(w_{1}) \ne f'(w_{1})$, and thus $w_{1} \in R_{2}$.
For $v \notin \{w_{1}, w_{2}\}$, we have
$v \in R$ if and only if $v \in R_{2}$ since $f_{2}(v) = f(v)$.
Hence, it holds that $R_{2} = R \mathbin{\dot{\cup}} \{w_{2}\}$.
Now the set $E_{D_{2}}$ is obtained from $E_{D}$ 
by removing $(f(w_{1}), f'(w_{1}))$ and 
adding $(f_{2}(w_{2}), f'(w_{2})) = (f(w_{1}), f'(w_{2}))$ and $(f_{2}(w_{1}), f'(w_{1}))= (f(w_{2}), f'(w_{1}))$.
Assume first that $f(w_{2}) \in V_{D}$. The same discussion as for the case of $w_2 \in R$ and $R = R_2$
shows that  $\mathrm{cc}(D_{2}) \ge \mathrm{cc}(D) - 1$:
if $f(w_{1})$ and $f(w_{2})$ are in the different component, $\mathrm{cc}(D_{2}) = \mathrm{cc}(D) - 1$;
otherwise $\mathrm{cc}(D_2) \ge \mathrm{cc}(D)$.
Next assume that $f(w_{2}) \notin V_{D}$, and thus $V_{D_{2}} = V_{D} \mathbin{\dot\cup} \{f(w_{2})\}$.
Since $f(w_{2}) = f'(w_{2})$,
$(f(w_{1}), f'(w_{2})) = (f(w_{1}),  f(w_{2}))$ and
$(f(w_{2}), f'(w_{1})) = (f'(w_{2}), f'(w_{1}))$ hold.
Hence, $D_{2}$ is obtained from $D$ by adding the vertex $f(w_{2})$ 
and replacing the edge $(f(w_{1}), f'(w_{1}))$ with 
the path $(f(w_1), f(w_2), f'(w_1))$.
Thus, $\mathrm{cc}(D_2) = \mathrm{cc}(D)$ holds.
Therefore, it holds in both cases that
\[
|R_{2}| + \mathrm{cc}(D_{2}) 
%= (|R| + 1) + \mathrm{cc}(D_{2})
\ge (|R| + 1) + \mathrm{cc}(D) - 1 
> |R| + \mathrm{cc}(D) - 1,
\]
as required.
\end{proof}

\subsection{The whole algorithm}

Let $\langle G=(V,E), f, f', k \rangle$ be an instance of \sts,
such that $G = (V,E)$ is a block-cactus graph with $|V| = n$ and $|E| = m$.
We first compute the set $\mathcal{B}_{G}$ of the biconnected components.
For each $H \in \mathcal{B}_{G}$, we mark all cut vertices, 
check whether $H$ contains a vertex $v$ with $f(v) \ne f(v')$,
and check whether $H$ is a cycle or a complete graph.
If $H$ is a complete graph, then we construct an implicit representation so that
we do not have to store the redundant information $E(H)$.
These preprocessing can be done in $O(m+n)$ time in total.
 
Let $s, t \in V$.
We compute the instance $I_{H} = \langle H, f_{H}, f'_{H}, s_{H}, t_{H}, P_{H} \rangle$ of \substs{} for all $H \in \mathcal{B}_{G}$.
We can do it in $O(m+n)$ time in a bottom-up manner over the tree structure of the biconnected components. 
Let $H \in \mathcal{B}_{G}$.
If $H$ is a cycle, then we compute $\lensubsts(I_{H})$ in $O(|V(H)|)$ time
using the algorithm in \cref{lem:subpd-cycle}.
If $H$ is a complete graph, then we compute $\lensubsts(I_{H})$ using \cref{eq:subpd-complete-graph},
which can be done in $O(|V(H)|)$ time from the implicit representation of $H$.
Thus, by \cref{lem:sub-pd-sum}, 
we can solve \ststs{$s$}{$t$} in $O(n)$ time, given that the aforementioned preprocessing is done.
Since we have $n^{2}$ candidates for the pair $s,t$, 
the total running time is $O(n^{3})$. This completes the proof of \cref{thm:block-cactus}.

Note that we only need $O(m+n)$ time to solve \ststs{$s$}{$t$}.
\begin{corollary}
\ststs{$s$}{$t$} on block-cactus graphs can be solved in linear time.
\end{corollary}

% complexity
\section{Hardness of the few-color and colorful cases}
\label{sec:hardness}

Since \sts{} clearly belongs to NP, in the following we only show the NP-hardness for each case.

\subsection{General tools for showing hardness}

The first tool uses the hardness of \textsc{Hamiltonian Path} 
to show the hardness of \sts{} with few colors.

A path (a cycle) in a graph is a \emph{Hamiltonian path} (a \emph{Hamiltonian cycle}, resp.)
if it visits every vertex in the graph exactly once.
The problems \textsc{Hamiltonian Path} and \textsc{Hamiltonian Cycle}
ask whether a given graph has a Hamiltonian path or a Hamiltonian cycle, respectively.
In the problem $(s,t)$-\textsc{Hamiltonian Path},
the first and last vertices $s$ and $t$ are fixed.

By \emph{attaching a cycle} of length $q$ at a vertex $v$,
we mean the operation of adding $q-1$ new vertices
and $q$ edges that form a cycle with $v$.

\begin{theorem}
\label{thm:few-color-gen}
Let $\mathcal{C}$ be a graph class.
For every fixed $c \ge 2$, \sts{} with $c$ colors is NP-complete on $\mathcal{C}$
if the following conditions are satisfied:
\begin{enumerate}
  \item $(s,t)$-\textsc{Hamiltonian Path} is NP-complete on $\mathcal{C}$;
  \item there is an integer $q \ge 3$ such that $\mathcal{C}$ is closed under the operation 
  that attaches a cycle of length $q$ at a vertex.
\end{enumerate}
\end{theorem}
\begin{proof}
We only prove the case of $c = 2$.
For a larger $c$, we can attach as many cycles as we need and use extra colors there.

Let $G = (V,E)$ be a graph in $\mathcal{C}$ and $s,t \in V$.
Let $H$ be the graph obtained from $G$ by attaching 
a cycle of length $q$ at each vertex.
Note that $H \in \mathcal{C}$ by the assumption.
We set $k = (q+1)(|V|-2) + 3$.

Now we define colorings $f$ and $f'$ (see \cref{fig:few-color-gen}).
We call the color 1 \emph{white} and the color 2 \emph{black}.
For each $v \in V$, let $C_{v}$ denote the cycle attached at $v$.
Let $f$ color one neighbor of $s$ in $C_{s}$ black,
and all other vertices in $C_{s}$ white.
We let $f'$ color all vertices in $C_{s}$ white.
Similarly, let $f$ color all vertices in $C_{t}$ white and 
let $f'$ color one neighbor of $t$ in $C_{t}$ black,
and all other vertices in $C_{t}$ white.
For $v \in V \setminus \{s,t\}$, let $f$ color one neighbor of $v$, say $w$, in $C_{v}$ black,
and all other vertices in $C_{v}$ white. Let $f'$ color the neighbor of $w$ in $V(C_{v}) \setminus \{v\}$ black,
and all other vertices in $C_{v}$ white.

For each $v \in V \setminus \{t\}$,
let $b_{v}$ denote the unique black vertex $u \in V(C_{v})$ under $f$.
Similarly, for each $v \in V \setminus \{s\}$,
let $b'_{v}$ denote the unique black vertex $u \in V(C_{v})$ under $f'$.
Note that for each $v \in V \setminus \{s,t\}$,
$b_{v}$ and $b'_{v}$ are adjacent.

\begin{figure}[tbh]
  \centering
  \includegraphics[scale=.9]{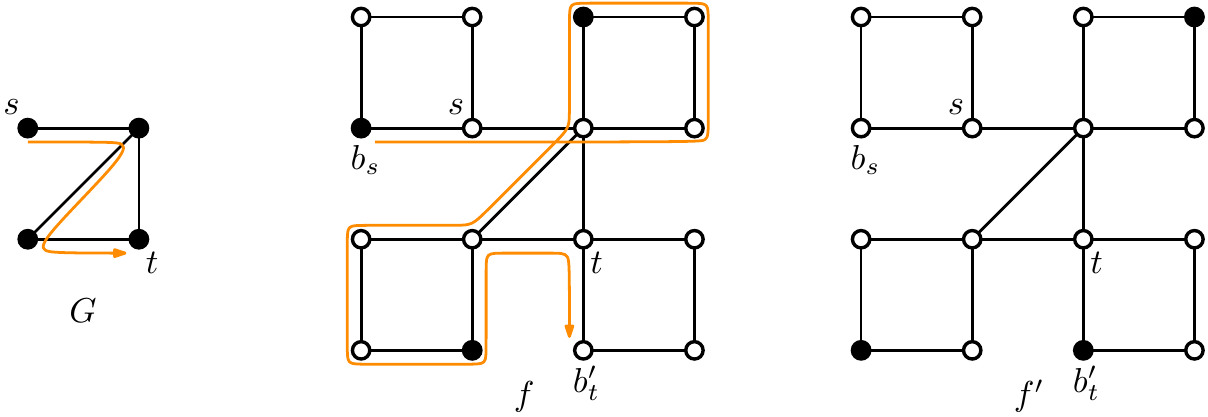}
  \caption{The reduction for \cref{thm:few-color-gen} ($q=4$).}
  \label{fig:few-color-gen}
\end{figure}

We show that $\langle G, s, t \rangle$ is a yes-instance of $(s,t)$-\textsc{Hamiltonian Path}
if and only if  
$\langle H, f, f', k \rangle$ is a yes-instance of \sts{}.

First assume that $G$ has a Hamiltonian path $P = \langle v_{1}, v_{2}, \dots, v_{n} \rangle$ from $v_{1} = s$ to $v_{n} = t$.
Let $W = \langle b_{s}, s, \overrightarrow{C}_{v_{2}}, \overrightarrow{C}_{v_{3}}, \dots, \overrightarrow{C}_{v_{n-1}}, t, b'_{t}  \rangle$
be a walk in $H$ such that 
for $2 \le i \le n-1$, $\overrightarrow{C}_{v_{i}}$ is the cyclic walk on $C_{v_{i}}$
such that it starts and ends at $v_{i}$ 
and proceeds in the ordering that $b_{v_{i}}$ is visited right before the second visit of $v_{i}$.
The swapping sequence along $W$ brings the black token on $b_{s}$ to $b'_{t}$
and moves the black token on $b_{v}$ to $b'_{v}$ for each $v \in V \setminus \{s,t\}$.
That is, the coloring obtained is~$f'$.
The length of the sequence is $4 + 2(|V|-2) + (q-1)(|V|-2) - 1 = (q+1)(|V|-2) + 3$
since it visits $b_{s}$, $s$, $t$, and $b'_{t}$ once,
each $v \in V \setminus \{s,t\}$ twice,
and all vertices in $V(C_{v}) \setminus \{v\}$ for $v \in V \setminus \{s,t\}$ once.

Next assume that $\langle H, f, f', k \rangle$ is a yes-instance of \sts{}.
Let $W$ be a walk corresponding to a swapping sequence of length at most $k$ from $f$ to $f'$.
Observe that $W$ has to start at $b_{s}$ and end at $b'_{t}$
since otherwise one cannot get rid of the black token in $C_{s}$
nor place a black token in $C_{t}$.
Since $f(b_{v}) \ne f'(b_{v})$ for each $v \in V \setminus \{s,t\}$,
$W$ has to visit $C_{v}$ for each $v \in V \setminus \{s,t\}$.
Furthermore, $W$ has to visit all vertices of $C_{v}$ for each $v \in V \setminus \{s,t\}$
since otherwise the moves in $C_{v}$ cancel out and thus the black token in $C_{v}$ cannot move.
Now consider the walk $W_{G}$ obtained from $W$ by taking the edges of $G$ only.
From the discussion above, we can bound the length of $W_{G}$ as
$|W_{G}|-1 \le |W|-1 - (1 + 1 + (|V|-2)q) \le k - (2 + (|V|-2)q) \le |V| - 1$.
Observe that $W_{G}$ starts at $s$, visits all vertices in $V \setminus \{s,t\}$, and ends at $t$.
This implies that $W_{G}$ is a Hamiltonian path of $G$ from $s$ to $t$.
\end{proof}

The second tool uses the hardness of \textsc{Steiner Tree} 
to show the hardness of the colorful case of \sts{}.
In this case, we ask $f$ to be injective
and call this condition the \emph{colorful condition}.

For a walk $W$, let $V(W)$ be the set of vertices in~$W$.
\begin{lemma}
\label{lem:non-changing-walk}
Let $f$ be an injective coloring of a graph $G$.
If a walk $W$ in $G$ corresponds to a swapping sequence from $f$ to $f$ itself,
then $|V(W)| \le (|W|+1)/2$.
\end{lemma}
\begin{proof}
We use induction on $|W|$. A walk $W$ with $|W| = 1$ satisfies the statement.

Assume that $|W| > 1$ and the statement holds for all strictly shorter walks.
If each vertex in $V(W)$ appears at least twice in $W$, then $|V(W)| \le |W|/2$ holds.
Thus assume that there exists a vertex $v \in V(W)$ that appears only once in $W$.
Since $f$ is injective, $W$ has to start and end at the same vertex,
and thus $v$ cannot be the first (and last) vertex in $W$.
Let $u$ and $w$ be the vertex right before and after $v$ in $W$, respectively.

When the swapping sequence along $W$ visits $v$, the token of unique color $f(v)$ moves to $u$.
The next move brings the token on $w$ to $v$, and after that the color of $v$ never changes.
This implies that indeed $u = w$.
Let $W'$ be the walk obtained from $W$ by replacing the subwalk $\langle u, v, u\rangle$ with $\langle u \rangle$.
Observe that $|W'| = |W|-2$ and $|V(W')| = |V(W)| - 1$.
By the induction hypothesis, $|V(W')| \le (|W|'+1)/2$ holds.
This implies that $|V(W)| - 1 \le ((|W|-2)+1)/2$,
and thus $|V(W)| \le (|W|+1)/2$.
\end{proof}

For a graph $G = (V,E)$ and a set $K \subseteq V$,
the subgraph $T$ of $G$ is a \emph{Steiner tree} if it is a tree and contains all vertices in $K$.
\begin{description}
  \item[Problem:] \textsc{Steiner Tree}
  \item[Input:] A graph $G = (V, E)$, a set $K \subseteq V$ with $|K| \ge 2$, and an integer $\ell$.
  \item[Question:] Is there a connected subgraph $T$ of $G$ such that $K \subseteq V(T)$ and $|E(T)| \le \ell $?
\end{description}

\begin{theorem}
\label{thm:colorful-gen}
\sts{} with the colorful condition is NP-complete on a graph class $\mathcal{C}$
if the following conditions are satisfied:
\begin{enumerate}
  \item \textsc{Steiner Tree} is NP-complete on $\mathcal{C}$;
  \item there is an integer $q \ge 3$ such that $\mathcal{C}$ is closed under the operation 
  that attaches a cycle of length $q$ at a vertex.
\end{enumerate}
\end{theorem}
\begin{proof}
Let $G = (V,E)$ be a graph in $\mathcal{C}$, $K \subseteq V$, and $\ell$ be an integer.
Let $H \in \mathcal{C}$ be the graph obtained from $G$ by attaching a cycle of length $q$ at each vertex in $K$.
For each $v \in K$, let $C_{v}$ denote the cycle attached at $v$.
We set $k = 2 \ell + q |K|$.
Let $f$ be an injective coloring of $G$
and $f'$ be the coloring obtained from $f$ by cyclically shifting one step (in an arbitrary direction) the colors of the vertices in $V(C_{v}) \setminus \{v\}$
for each $v \in K$. See \cref{fig:colorful-gen}.

\begin{figure}[tbh]
  \centering
  \includegraphics[scale=.9]{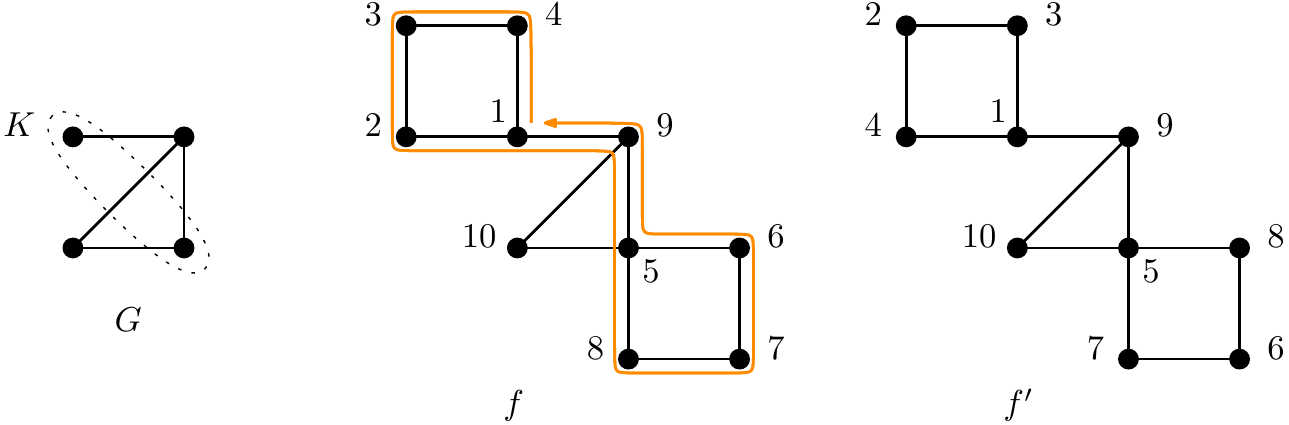}
  \caption{The reduction for \cref{thm:colorful-gen} ($q=4$).}
  \label{fig:colorful-gen}
\end{figure}

We show that $\langle G, K, \ell \rangle$ is a yes-instance of \textsc{Steiner Tree}
if and only if  
$\langle H, f, f', k \rangle$ is a yes-instance of \sts{}.

\paragraph{The only-if direction.}
First assume that there is a tree $T$ such that $T$ is a subgraph of $G$, $K \subseteq V(T)$, and $|E(T)| \le \ell$.
Let $W$ be a walk on $T$ that visits each edge of $T$ exactly twice.
For each $v \in K$, we expand $W$ by inserting a walk right after the first occurrence of $v$ in $W$
such that the inserted walk is a cycle through $C_{v}$ that 
has direction opposite to the cyclic shift applied to $C_{v}$ when we constructed $f'$.
We can see that the swapping sequence along this walk obtains $f'$.
The length of the sequence is $2 |E(T)| + q |K| \le 2\ell + q |K| = k$.

\paragraph{The if direction.}
Next assume that $\langle H, f, f', k \rangle$ is a yes-instance of \sts{}.
Let $W$ be a walk corresponding to a swapping sequence of length at most $k$ from $f$ to $f'$.
Let $R$ be the set of vertices that $W$ visits.
From $W$, we construct a walk $W'$ by taking the vertices of $G$ only.

If $W$ does not start in $V(C_{v}) \setminus \{v\}$ for some $v \in K$,
then $W$ contains a cyclic subwalk that visits $v$, $V(C_{v}) \setminus \{v\}$, and $v$ again in this ordering.
In $W'$, this cyclic subwalk is replaced with the trivial walk $\langle v \rangle$.
This decreases the length of the walk by $q$.

If $W$ starts at a vertex in $V(C_{v}) \setminus \{v\}$ for some $v \in K$,
it has to end at a vertex in $V(C_{v}) \setminus \{v\}$.
We claim that $W$ contains the edges of $C_{v}$ at least $q$ times in total.
Suppose to the contrary that this is not the case.
Then, there is an edge $e$ in $C_{v}$ that is not included in $W$.
Since $W$ has to visit all vertices in $V(C_{v}) \setminus \{v\}$ as they have different colors in $f$ and $f'$,
$W$ contains all other edges in $C_{v}$.
This implies that $W$ starts at an endpoint of the missing edge $e$,
leaves $C_{v}$, visits some other vertices and comes back to $C_{v}$,
and then proceeds in $C_{v}$ and ends at the other endpoint of $e$ (see \cref{fig:colorful-gen-bad}).
In this case, some vertex in $V(C_{v}) \setminus \{v\}$ will be colored with $f(v)$ in the final coloring.
This contradicts the assumption that $f$ is injective and $f(v) = f'(v)$.
Hence, this cycle also contributes to the shrinking of the length of the walk by at least $q$.

\begin{figure}[tbh]
  \centering
  \includegraphics[scale=.9]{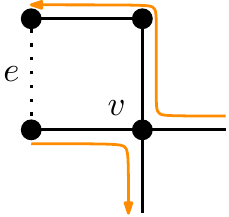}
  \caption{A walk contains the edges of $C_{v}$ at least $q$ times in total.}
  \label{fig:colorful-gen-bad}
\end{figure}

The discussions above imply that $|W'| \le |W| - q|K|$.
Since $|W| \le k+1 = 2 \ell+q|K| + 1$, we have $|W'| \le 2 \ell + 1$.
Observe that $W'$ is a swapping sequence from $f$ to $f$ itself.
By \cref{lem:non-changing-walk}, it holds that $|V(W')| \le (|W'|+1)/2 \le \ell + 1$.
Since $G[V(W')]$ is connected, it has a spanning tree of at most $\ell$ edges.
\end{proof}

It is known that \textsc{Steiner Tree} is NP-complete on 
chordal graphs~\cite{WhiteFP85} and
chordal bipartite graphs~\cite{MullerB87}.
It is also known that $(s,t)$-\textsc{Hamiltonian Path} is NP-complete on chordal graphs and chordal bipartite graphs~\cite{Muller96a}.
Observe that chordal graphs and chordal bipartite graphs are closed under the operations 
that attach a cycle of length $3$ and $4$, respectively.
Thus, \cref{thm:few-color-gen,thm:colorful-gen} implies the hardness on them.
\begin{corollary}
\sts{} is NP-complete on chordal graphs and on chordal bipartite graphs in 
both the colorful and few-color cases.
\end{corollary}

\subsection{The few-color case on grid-like graphs}

Recall that a graph is a grid graph if it is an induced subgraph of a grid.
A bipartite graph is \emph{balanced} if it admits a proper 2-coloring 
such that the color classes have the same size.
Note that a grid graph is bipartite.

It is known that \textsc{Hamiltonian Cycle} is NP-complete on grid graphs~\cite{ItaiPS82}.
The next lemma follows easily from this fact.
\begin{lemma}
\label{lem:balanced-ham-path}
\textsc{Hamiltonian Path} is NP-complete on balanced grid graphs given with grid representations.
\end{lemma}
\begin{proof}
We show the NP-hardness by giving a reduction from 
\textsc{Hamiltonian Cycle} on grid graphs given with grid representations,
which is known to be NP-complete~\cite{ItaiPS82}.
Let $G = (V, E)$ be an instance of this problem.
We assume that $G$ is balanced and has minimum degree at least $2$
as otherwise $G$ does not have any Hamiltonian cycle.

Let $v = (m_{x}, m_{y}) \in V$ be the top vertex in the rightmost column;
that is, $m_{x} = \max\{x \mid (x,y) \in V\}$ and $m_{y} = \max\{y \mid (m_{x},y) \in V\}$.
By the assumption of the minimum degree, $v$ has degree~$2$.
Let $u = v + (-1,0)$ and $w = v + (0,-1)$ be the neighbors of $v$ (see \cref{fig:ham-c-p} (left)).
Let $H$ be the graph obtained from $G$ by adding four vertices
$s = v+(1,1)$, $s' = v+(1,0)$, $t = v+(2,-1)$, and $t' = v+(1,-1)$ into the grid representation
(see \cref{fig:ham-c-p} (right)).
Clearly, $H$ is balanced.

Observe that every Hamiltonian cycle of $G$ (if any exists) contains the path $\langle u,v,w \rangle$, and
every Hamiltonian path of $H$ (if any exists) contains the paths $\langle s, s', v, u\rangle$ 
and $\langle w, t', t\rangle$.
Since a Hamiltonian cycle of $G$ and a Hamiltonian path of $H$
play the same role in the subgraphs induced by $V \setminus \{v\}$,
we can conclude that $G$ has a Hamiltonian cycle if and only if $H$ has a Hamiltonian path.
\begin{figure}[tbh]
  \centering
  \includegraphics[scale=.85]{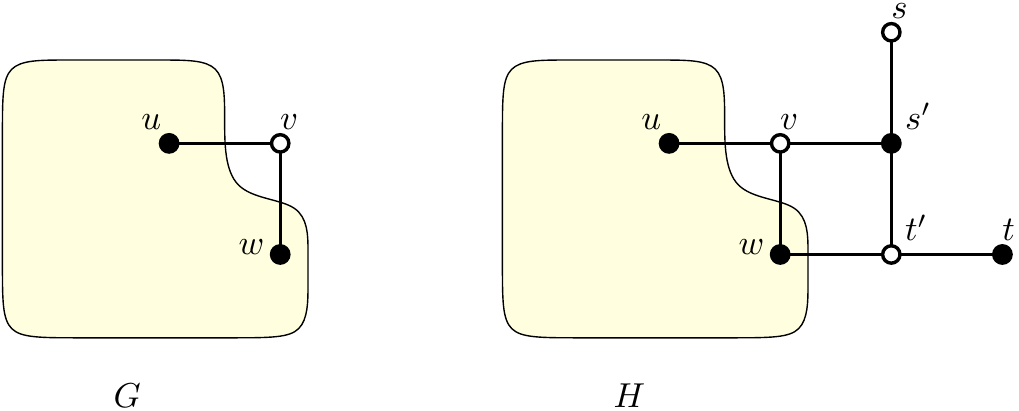}
  \caption{The graphs $G$ and $H$ in the proof of \cref{lem:balanced-ham-path}.}
  \label{fig:ham-c-p}
\end{figure}
\end{proof}

\begin{theorem}
\label{thm:kings-pd-even-constant-number-colors-np-complete}
For every fixed constant $c \ge 2$,
\sts{} with $c$ colors is NP-complete on king's graphs.
\end{theorem}
\begin{proof}
We prove the theorem only for the case where $c = 2$.
For $c > 2$, we add $c-2$ new vertices to $G$ defined below
and for each new vertex, set a new color as its initial and target colors. Then the proof works as it is.

We prove the NP-hardness by a reduction from \textsc{Hamiltonian Path} on balanced grid graphs 
(see \cref{lem:balanced-ham-path}).

Let $G = (V, E)$ be a balanced grid graph given with a grid representation.
We assume that $G$ is connected.
From $G$, we construct an instance $\langle H, f, f', k \rangle$ of \sts.
We set $k = |V|-1$.

Let $\min_{\mathbf{x}} = \min\{x \in \mathbb{Z} \mid (x, y) \in V\}$
and $\min_{\mathbf{y}} = \min\{y \in \mathbb{Z} \mid (x, y) \in V\}$.
We also define $\max_{\mathbf{x}}$ and $\max_{\mathbf{y}}$ in the analogous ways.
Let $U = \{(x,y) \in \mathbb{Z}^{2} \setminus V \mid \min_{\mathbf{x}} \le x \le \max_{\mathbf{x}}, \,
 \min_{\mathbf{y}} \le x \le \max_{\mathbf{y}}\}$.
The grid graph represented by $U \cup V$ is a grid and has size polynomial in $|V|$.
From this grid, we obtain $H$ by adding all diagonal edges of the unit squares.
Note that $H$ is a king's graph.

Let $f$ be the coloring of $H$ that maps the odd vertices to $1$ and the even vertices to $2$.
Let $f'$ be the coloring obtained from $f$ by reversing the colors of the vertices in the original grid graph $G$.
That is, 
$f'(v) = f(v)$ for $v \in U$,
$f'(v) = 1$ for $v \in V$ with $f(v) =2$, and
$f'(v) = 2$ for $v \in V$ with $f(v) =1$.

We show that $G$ has a Hamiltonian path
if and only if
$\langle H, f, f', k \rangle$ is a yes-instance of \sts.

\paragraph{The only-if direction.}
Let $P = \langle v_{1}, \dots, v_{|V|} \rangle$ be a Hamiltonian path of $G$.
Let $S = \langle f_{1}, \dots, f_{|V|} \rangle$ be a swapping sequence corresponding to $P$, where $f_{1} = f$.
Since each vertex in the walk is visited only once,
we have $f_{|V|}(v_{i}) = f(v_{i+1})$ for $1 \le i \le |V|-1$,
and $f_{|V|}(v_{|V|}) = f(v_{1})$.
Since $P$ is a path of $G$, $v_{i}$ and $v_{i+1}$ have different parities.
Also, since $G$ is balanced, $v_{|V|}$ and $v_{1}$ have different parities.
Therefore, $f_{|V|}$ is obtained from $f_{1}$ by changing the color of each vertex in $V$ to the other one.
That is, $f_{|V|} = f'$.

\paragraph{The if direction.}
Let $\langle f_{1}, \dots, f_{k'+1} \rangle$ be a swapping sequence between $f$ and $f'$ with $k' \le k$.
Let $W = \langle w_{1}, \dots, w_{k'+1} \rangle$ be the corresponding walk in $H$.
Since $f(v) \ne f'(v)$ for every $v \in V$, the moving token has to visit all vertices in $V$.
Furthermore, since $|W| = k'+1 \le k+1 = |V|$,
indeed the moving token visits each vertex in $V$ exactly once and does not visit other vertices (and thus, $k' = k$).
Hence, it suffices to show that $W$ is a walk also in $G$;
that is, each edge $\{w_{i}, w_{i+1}\}$ is not diagonal.
Suppose to the contrary that an edge $\{w_{i}, w_{i+1}\}$ in $W$ is diagonal.
This implies that $w_{i}$ and $w_{i+1}$ have the same parity, and thus $f(w_{i}) = f(w_{i+1})$.
Since $W$ visits a vertex at most once, $f(w_{i+1}) = f'(w_{i})$ has to hold.
Thus we have that $f(w_{i}) = f'(w_{i})$.
This contradicts the assumption that $f'(v) \ne f(v)$ for each $v \in V$.
\end{proof}

In the proof above,
we showed that no diagonal edge is used in shortest swapping sequences.
Therefore, the proof works without the diagonal edges.
\begin{corollary}
For every fixed constant $c \ge 2$,
\sts{} with $c$ colors is NP-complete on grids.
\end{corollary}

Observe further that the proof of \cref{thm:kings-pd-even-constant-number-colors-np-complete} works
even if we add or remove an arbitrary set of edges connecting vertices of the same parity 
since we can just ignore them.
Now consider the graph obtained from a king's graph by removing all edges connecting even vertices
and adding all possible edges connecting odd vertices.
Such a graph is a split graph since the even vertices form an independent set and the odd vertices form a clique.
Thus it is hard on split graphs as well.
\begin{corollary}
For every fixed constant $c \ge 2$,
\sts{} with $c$ colors is NP-complete on split graphs.
\end{corollary}

\subsection{The colorful case on grid-like graphs}

We now consider \sts{} with the colorful condition on grid-like graphs.

We first show the hardness on the ordinary grids.
Recall that \ststs{$s$}{$t$} with the colorful condition on grids are known as the generalized 15 puzzle (or the $(n^{2}-1)$ puzzle)
and shown to be NP-complete~\cite{RatnerW90,DemaineR18}. 
For \sts{}, we can use the reduction by Demaine and Rudoy~\cite{DemaineR18} almost directly with a small change.
Their reduction is from the following problem.
\begin{description}
  \item[Problem:] \textsc{Rectilinear Steiner Tree}
  \item[Input:] A set $P \subseteq \mathbb{Z}_{+}^{2}$ of integer points in the plane and an integer $\ell$.
  \item[Question:] Is there a tree $T$ on the plane that satisfies the following conditions?
    $T$ contains all points in $P$;
    every edge of $T$ is horizontal or vertical;
    the total length of the edges in $T$ is at most $\ell$.
\end{description}
\textsc{Rectilinear Steiner Tree} is known to be strongly NP-hard~\cite{GareyJ77},
and thus we assume that the maximum coordinate of the points in $P$
is bounded from above by a polynomial in $|P|$.

The high-level idea of the reduction in~\cite{DemaineR18} is to represent the integer points in the plane by a grid 
and then each point in $P$ by some local changes of the colors around the vertex corresponding to the point.
Then, a swapping sequence for this instance forms a rectilinear Steiner tree on the plane.
The difference between their setting and ours is that they can fix the starting and ending vertices, but we cannot.
This difference actually does not affect the correctness of their proof applied to our case.
To not repeat their argument here, we only give a proof sketch.

\begin{theorem}
\label{thm:grid-pd-all-color-different-np-complete}
\sts{} with the colorful condition is NP-complete on grids.
\end{theorem}
\begin{proof}
[sketch]
From an instance of $\langle P, \ell \rangle$ \textsc{Rectilinear Steiner Tree},
we construct an instance $\langle G, f, f', k \rangle$ of \sts{}.

Let $c = 18|P|$ and $D$ be the maximum coordinate of the points in $P$.
Let $G$ be the $((D+1) c) \times ((D+1) c)$ grid.
Let $f$ be an injective coloring of $G$,
and $f'$ be the coloring obtained from $f$ 
by rotating the colors of vertices $(cx, cy)$, $(cx+1, cy)$, and $(cx, cy+1)$ as shown
in \cref{fig:grid-pd-1} for each point $(x, y) \in P$.
We set $k = (2\ell+1) c$.

This construction is almost the same as the one in~\cite{DemaineR18}.
In their setting, they fix one vertex $p_{1} \in P$ arbitrarily and 
set both $s$ and $t$ to the vertex corresponding to $p_{1}$.
They omit the color rotation around the vertices corresponding to $p_{1}$
as it is anyway visited as $s$ and $t$.

\begin{figure}[htb]
  \centering
  \includegraphics[scale=0.7]{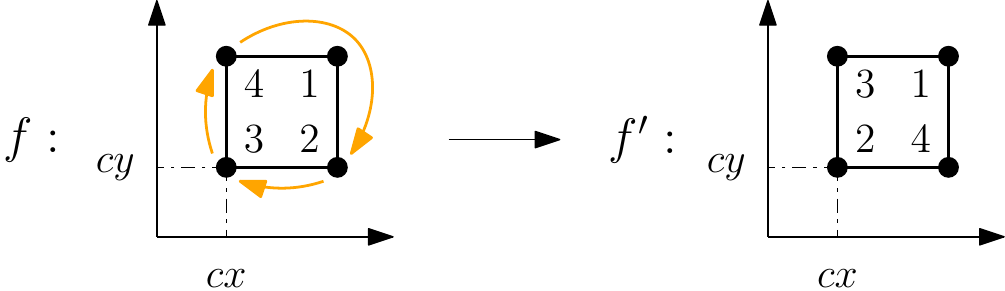}
  \caption{Rotating the colors around $(cx,cy)$.}
  \label{fig:grid-pd-1}
\end{figure}

First assume that $\langle P, \ell \rangle$ is a yes-instance of \textsc{Rectilinear Steiner Tree}.
Demaine and Rudoy~\cite{DemaineR18} showed that in this case there exists a swapping sequence $S$ such that:
\begin{itemize}
  \item the length is at most $2c \ell + 18(|P|-1)$;
  \item the moving token starts and ends at the vertex $(c x_{1}, c y_{1})$, where $(x_{1}, y_{1})$ is the coordinate of
  an arbitrarily chosen point $p_{1} \in P$;
  \item the obtained coloring is the same as $f$ for three vertices $(c x_{1}, c y_{1})$, $(c x_{1}+1, cy_{1})$ and $(cx_{1}, cy_{1}+1)$;
  \item the obtained coloring is the same as $f'$ for the other vertices.
\end{itemize}
We want to expand $S$ by inserting a subsequence that rotates
the colors of $(c x_{1}, c y_{1})$, $(c x_{1}+1, cy_{1})$, and $(cx_{1}, cy_{1}+1)$
and leaves the colors of the other vertices unchanged.
This is possible with at most $18$ steps~\cite{DemaineR18},
and thus we get a swapping sequence from $f$ to $f'$ with length at most $2c \ell + 18|P| = k$.

Next assume that $\langle G, f, f', k \rangle$ is a yes-instance of \sts{}.
For this direction, we can use the proof in~\cite{DemaineR18} as it is.
In the $1/c$ scale, we embed the edges used in a walk corresponding to a swapping sequence into the plane.
Then, using the analysis in~\cite{DemaineR18}, 
we can show that this walk (or a spanning tree of it) gives a rectilinear Steiner tree for $P$ with length at most $\ell$.
\end{proof}

Before showing the hardness of \sts{} with the colorful condition on king's graphs,
we first show that \textsc{Steiner Tree} is NP-complete on king's graphs
since the proofs are similar and this one is easier than the one for \sts.
Also, the result itself might be useful for connecting some graph problems and geometric problems.

\begin{theorem}
\label{thm:kings-steiner-np-complete}
\textsc{Steiner Tree} is NP-complete on king's graphs.
\end{theorem}
\begin{proof}
The problem clearly belongs to NP\@.
We prove the NP-hardness by giving a reduction from \textsc{Rectilinear Steiner Tree}.
From an instance $\langle P, \ell \rangle$ of \textsc{Rectilinear Steiner Tree},
we construct an instance $\langle G, K, \ell \rangle$ of \textsc{Steiner Tree}.
Let $D$ be the maximum coordinate of the points in $P$.
Let $G$ be the $(2D - 1) \times (2D - 1)$ king's graph.
For each point $(x, y) \in \{1, 2, \dots, D\}^2$, let $g((x, y)) = (x+y-1, -x+y+D)$.
We set $K = \{g(p) \mid p \in P\}$.
That is, we rotate the original instance $45$ degrees and embed it into 
the $(2D-1) \times (2D-1)$ king's graph (see \cref{fig:steiner-reduction-grid-kings}).
Note that the $(2D-1) \times (2D-1)$ king's graph contains the $D \times D$ grid as a subgraph.
We show that $\langle P, \ell \rangle$ is a yes-instance if and only if so is $\langle G, K, \ell \rangle$.
\begin{figure}[tbh]
  \centering
  \includegraphics[scale=0.7]{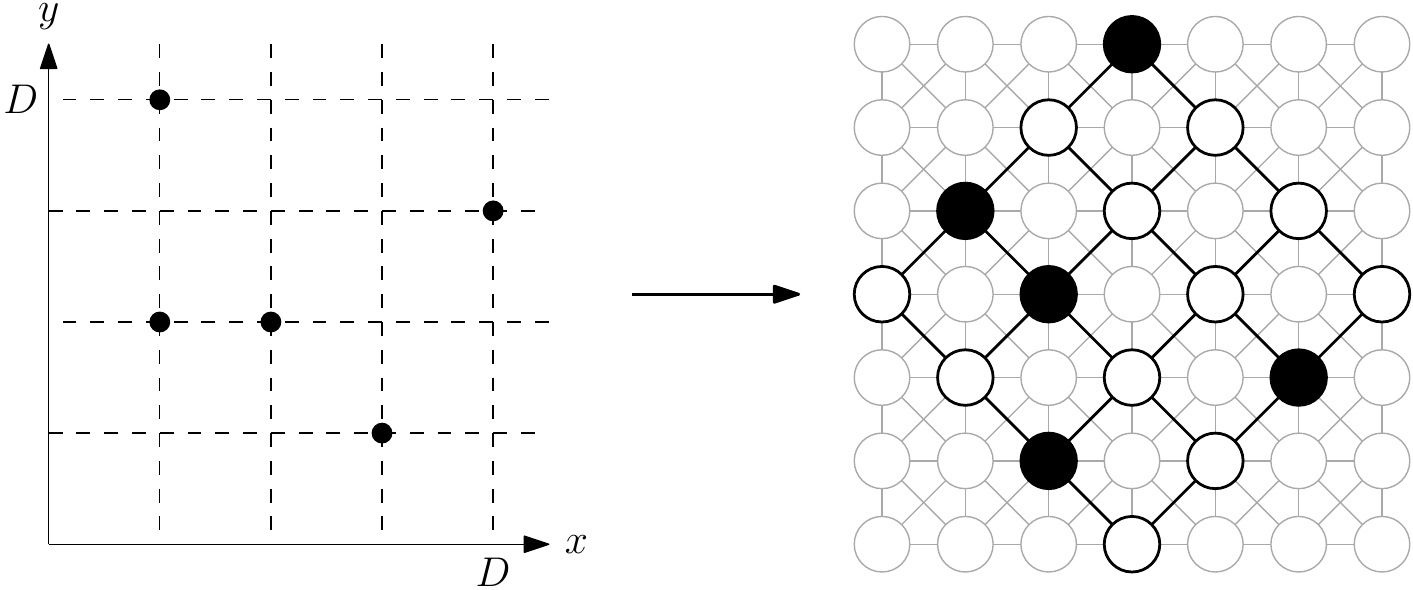}
  \caption{The reduction to \textsc{Steiner Tree} on king's graphs.}
  \label{fig:steiner-reduction-grid-kings}
\end{figure}

\paragraph{The only-if direction.}
Assume that $\langle P, \ell \rangle$ is a yes-instance of \textsc{Rectilinear Steiner Tree}.
Hanan~\cite{Hanan66} showed that there exists a minimum rectilinear Steiner tree for $P$
such that all edges are on a grid formed by taking the horizontal and vertical lines through each point in $P$.
This implies that 
by rotating $45$ degrees and scaling by a factor of $\sqrt{2}$,
we can embed $T$ into the $D \times D$ grid contained in $G$
and obtain a Steiner tree for $K$ with length at most $\ell$.
See \cref{fig:steiner-reduction-grid-kings_2}.
\begin{figure}[tbh]
  \centering
  \includegraphics[scale=0.7]{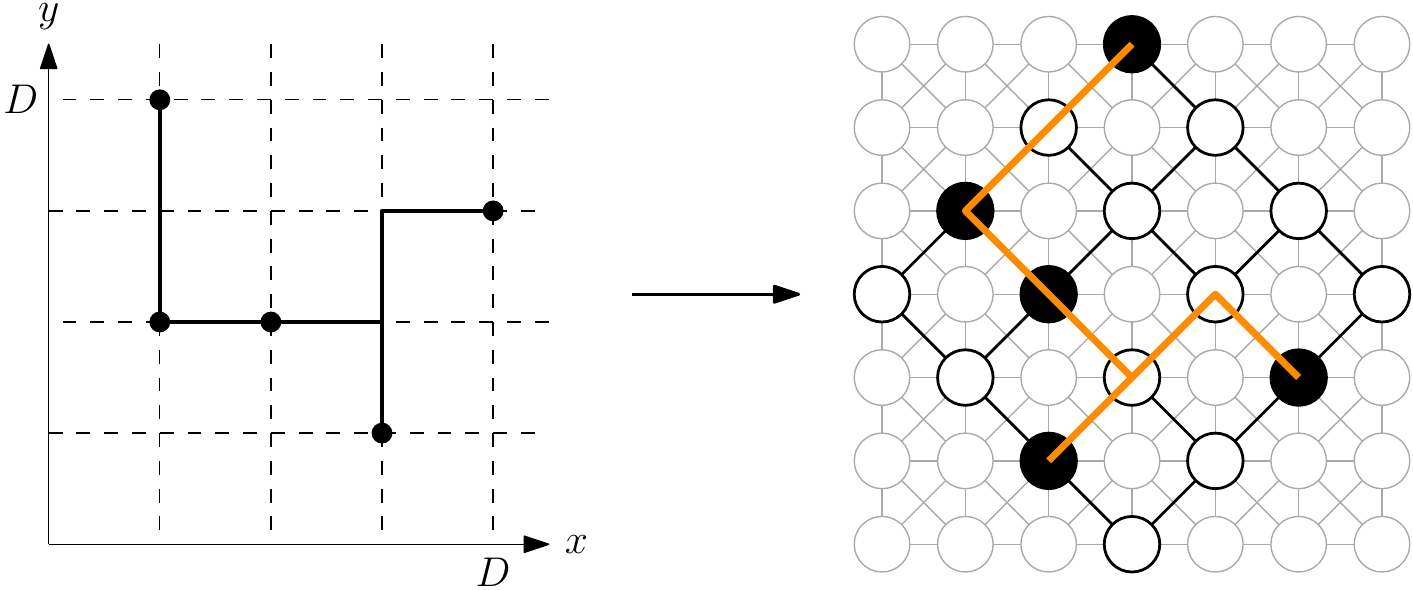}
  \caption{Embedding a rectilinear Steiner tree in the plane into a king's graph.}
  \label{fig:steiner-reduction-grid-kings_2}
\end{figure}

\paragraph{The if direction.}
Let $T$ be a Steiner tree for $K$ with length at most $\ell$.
By rotating $T$ $45$ degrees and scaling by a factor of $1 / \sqrt{2}$,
we obtain a (not necessarily rectilinear) tree $T'$ in the plane that includes all points in $P$
(see \cref{fig:steiner-reduction-from-kings-grid-restoration}).

\begin{figure}[tbh]
  \centering
  \includegraphics[scale=0.7]{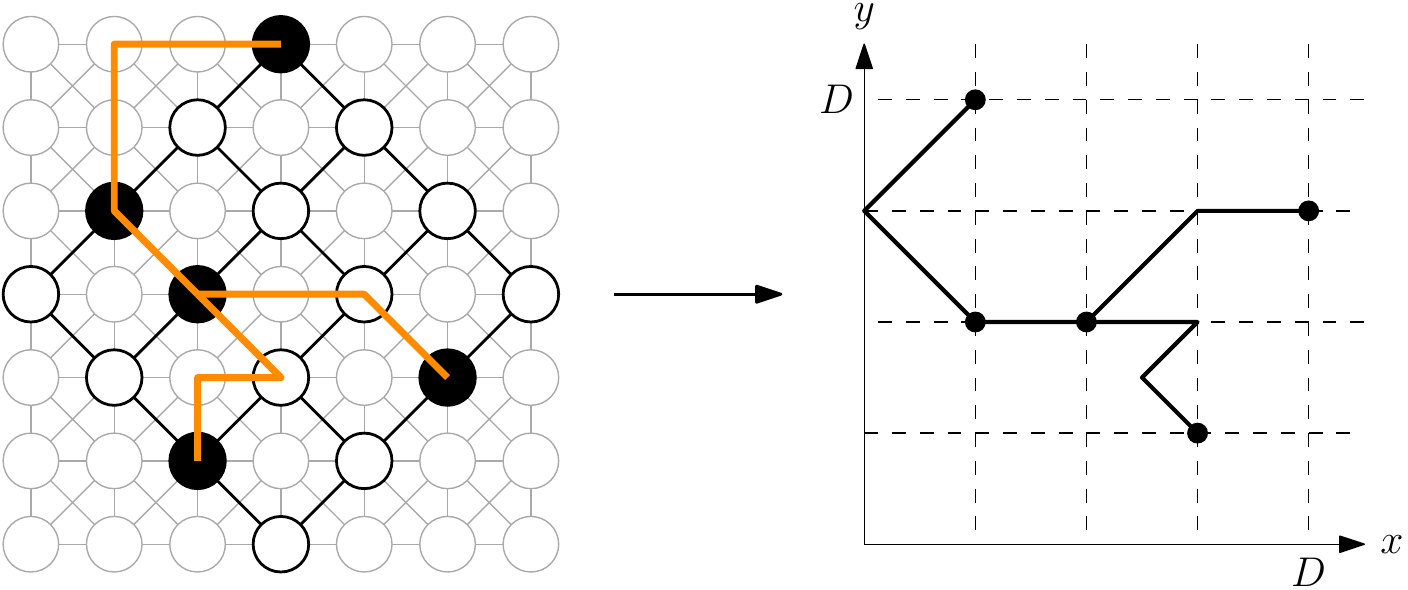}
  \caption{Embedding a Steiner tree of a king's graph into the plane.}
  \label{fig:steiner-reduction-from-kings-grid-restoration}
\end{figure}

We now replace the non-rectilinear edges in $T'$ with rectilinear edges, while maintaining the connectivity, as follows.
We partition each non-rectilinear edges in $T'$ into segments of length $1/\sqrt{2}$ 
(see \cref{fig:steiner-reduction-from-kings-grid-restoration-2} (left)).
Observe that each of such segments takes a form of $\{(x, y), (x+1/2, y+1/2)\}$ or $\{(x, y), (x+1/2, y-1/2)\}$.

We replace each segment $\{(x, y), (x+1/2, y+1/2)\}$ with two segments $\{(x, y), (x + 1/2, y)\}$ and $\{(x + 1/2, y), (x + 1/2, y + 1/2)\}$.
Similarly, we replace $\{(x, y), (x+1/2, y-1/2)\}$ with $\{(x, y), (x + 1/2, y)\}$ and $\{(x + 1/2, y), (x + 1/2, y - 1/2)\}$.
See \cref{fig:steiner-reduction-from-kings-grid-restoration-2} (right).
After the replacements, we take a spanning tree, which is a rectilinear Steiner tree for $P$, and call it $T''$.
\begin{figure}[tbh]
  \centering
  \includegraphics[scale=0.7]{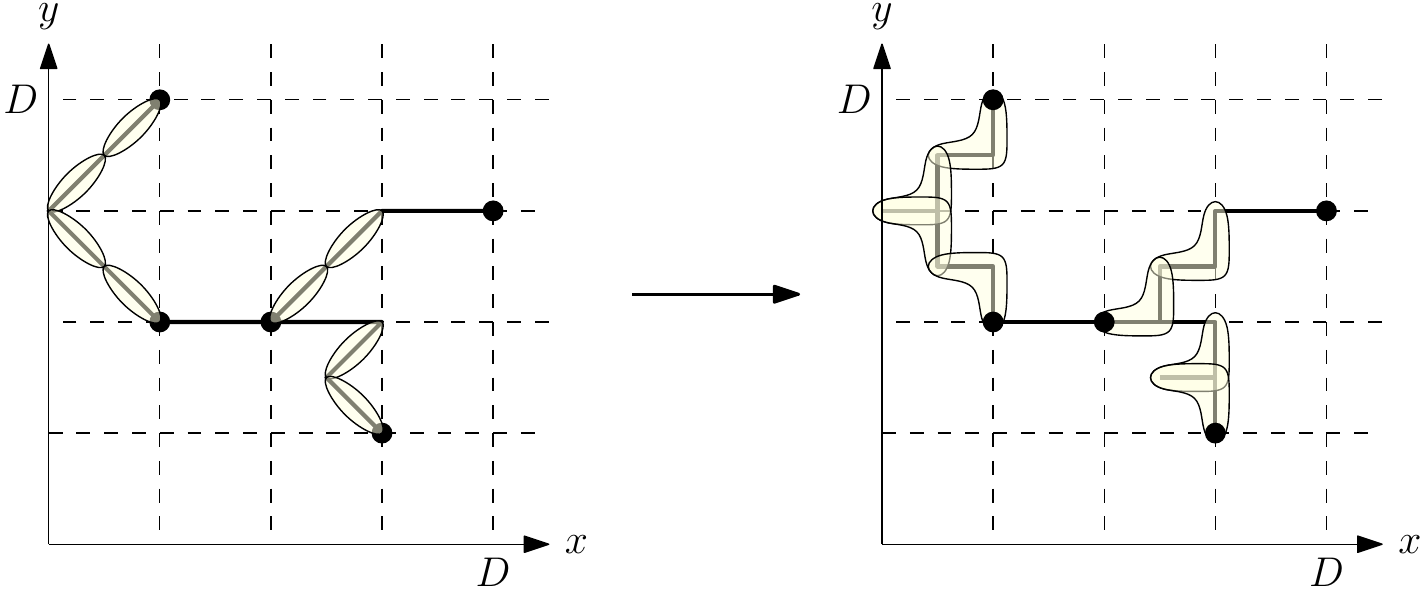}
  \caption{Replacing non-rectilinear edges.}
  \label{fig:steiner-reduction-from-kings-grid-restoration-2}
\end{figure}

Now we prove that the length of the rectilinear Steiner tree obtained above is at most $\ell$.
If an edge of $T$ is between vertices of the same parity, 
then the edge corresponds to a unit-length rectilinear segment in $T'$ (and thus in $T''$ as well).
Otherwise, the edge corresponds to two segments of length $1/2$ in $T''$.
Therefore, the length of the rectilinear Steiner tree is at most $|E(T)|$ ($\le \ell$).
\end{proof}

We now show the hardness of \sts{} with the colorful condition on king's graphs.
Although the proof is similar to the one for grids,
the presence of diagonal edges makes it a little more complicated.

\begin{theorem}
\label{thm:kings-pd-all-color-different-np-complete}
\sts{} with the colorful condition is NP-complete on king's graphs.
\end{theorem}
\begin{proof}
We prove the NP-hardness by a reduction from \textsc{Rectilinear Steiner Tree}.
From an instance $\langle P, \ell \rangle$ of \textsc{Rectilinear Steiner Tree},
we construct an instance $\langle G, f, f', k \rangle$ of \textsc{Steiner Tree} that satisfies the colorful condition.

Let $c = 12|P|$ and $D$ be the maximum coordinate of the points in $P$.
We set $k = (2 \ell + 1) c$.
Let $G$ be the $((2D+1)c) \times ((2D+1)c)$ king's graph.
For each point $p = (x, y) \in \{1, 2, \dots, D\}^2$, let 
$g(p) = (g_{x}(p), g_{y}(p)) = (c(x+y), c(-x+y+D+1))$.
Let $f$ be an injective coloring of $G$, and $f'$ be the coloring obtained from $f$ by rotating the colors 
of the four vertices $(g_x(p) \pm  1, g_y(p)), (g_x(p), g_y(p) \pm 1)$ for each $p \in P$ as shown in \cref{fig:kings-pd-reduction-rotation}.
We show that $\langle P, \ell \rangle$ is a yes-instance 
if and only if so is $\langle G, f, f', k \rangle$.

\begin{figure}[tbh]
  \centering
  \includegraphics[scale=0.7]{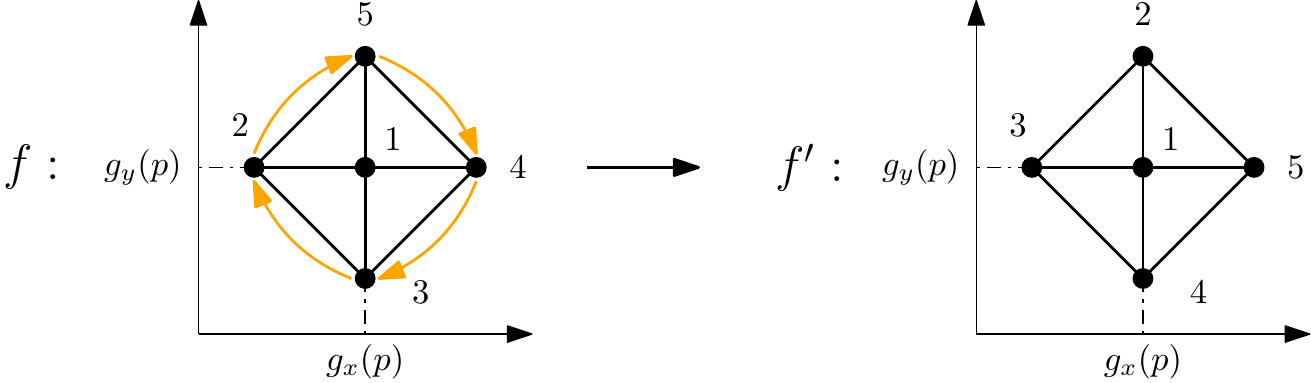}
  \caption{Rotating the colors around the vertex $g(p)$.}
  \label{fig:kings-pd-reduction-rotation}
\end{figure}

\paragraph{The only-if direction.}

Assume that $\langle P, \ell \rangle$ is a yes-instance of \textsc{Rectilinear Steiner Tree}.
By the same discussion as in \cref{thm:kings-steiner-np-complete}, 
there exists a Steiner tree for $P$ with at most $\ell$ edges in the $D \times D$ grid.
Let $T$ be such a tree.
Now, by scaling $T$ by a factor of $\sqrt{2} \cdot c$ and then rotating $45$ degrees, 
we obtain a Steiner tree $T'$ for $g(P)$ in $G$ such that the number of edges is at most $c \ell$.
That is, if $T$ contains an edge $\{u,v\}$,
then $T'$ contains the unique $g(u)$--$g(v)$ path of length $c$ in $G$ that corresponds to $\{u,v\}$.
Note that since $T$ is rectilinear, $T'$ only contains diagonal edges.

From $T'$, we construct a swapping sequence from $f$ to $f'$ with length at most $k$.
Let $p_{1}$ be an arbitrary vertex in $P$, and 
let $W$ be a walk from $g(p_{1})$ in $T'$ that contains each edge of $T'$ exactly twice.
As observed in~\cite{DemaineR18}, the moves along such a walk cancel out, and thus
the swapping sequence corresponding to $W$ is a swapping sequence from $f$ to $f$ itself.

We construct a swapping sequence between $f$ and $f'$ by adding some moves to $W$.
For each point $p \in P$, we insert five moves to $W$ just after the moving token visits $g(p)$ for the first time.
Suppose that the current configuration is $f$ and the moving token is at $g(p)$.
Then, moving the moving token as in \cref{fig:kings-pd-reduction-rotation-swapping-sequence}, we can match the colors of tokens at four vertices $(g_x(p) \pm  1, g_y(p) \pm 1)$ to $f'$.
Note that this movement does not affect the colors of other tokens. 
\begin{figure}[htb]
  \centering
  \includegraphics[width=\textwidth]{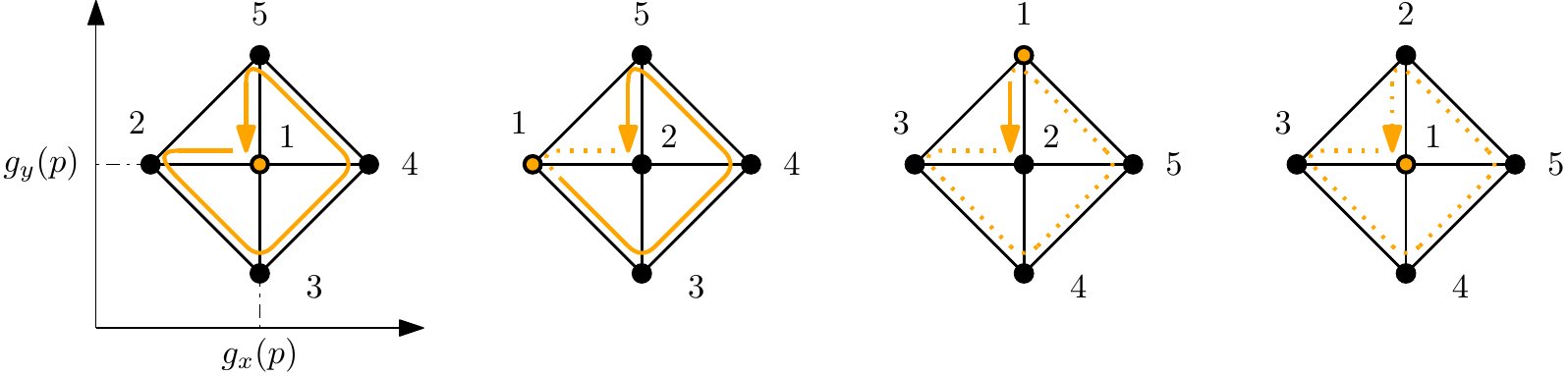}
  \caption{A movement to match the colors of tokens around $g(p)$ to $f'$.}
  \label{fig:kings-pd-reduction-rotation-swapping-sequence}
\end{figure}
Since $W$ contains all edges of $T'$, $W$ contains $g(p)$ for every $p \in P$.
In addition, the vertices in $T'$ have the same parity
since it contains only diagonal edges.
Thus, the four vertices $(g_x(p) \pm  1, g_y(p)), (g_x(p), g_y(p) \pm 1)$ do not appear in $W$.
This implies that the four vertices are touched by the inserted moves only.
Therefore, by inserting the five moves into $W$ for each $p \in P$, we obtain a swapping sequence between $f$ and $f'$.

Now we prove that the length is at most $k$.
Since $T'$ has at most $c \ell$ edges and $W$ contains each of them exactly twice, $|W| \le 2c \ell$ holds.
Since we add five moves to $W$ for each $p \in P$,
the length of the resultant walk is $|W| + 5|P| < 2c \ell + c = k$.

\paragraph{The if direction.}
Assume that $\langle G, f, f', k \rangle$ is a yes-instance of \sts{}.
As observed in the proof of \cref{thm:kings-steiner-np-complete}, 
if there is a Steiner tree for $g(P)$ on $G$ with $L$ edges, 
then there is a rectilinear Steiner tree for $P$ with length at most $L/c$.
Since the length of a minimum rectilinear Steiner tree is an integer~\cite{Hanan66} (see also \cite{DemaineR18}),
it suffices to show that there exists a Steiner tree for $g(P)$ with at most $(\ell + 1)c-1$ edges,
or, equivalently with at most $(\ell + 1)c$ vertices.

Let $S$ be a minimum swapping sequence from $f$ to $f'$ with length at most $k = (2\ell+1)c$.
Let $R$ be the set of vertices that the moving token visits in $S$,
and let $R' = \{v \in R \mid f(v) = f'(v)\}$.
If $f(v) \ne f'(v)$ for some $v$, then the moving token has to visit $v$.
Thus, $|R \setminus R'| = |\{v \in V(G) \mid f(v) \ne f'(v)\}| = 4 |P|$.
Since $S$ has the minimum length, each vertex in $R'$ is visited at least twice (see~\cite{DemaineR18}),
and thus $|R'| \le (k +1)/2$.
Therefore, it holds that
$|R| = |R'| + |R \setminus R'| \le (k +1)/2 + 4|P| = ((2 \ell + 1)c +1)/2 + 4|P| < (2 \ell + 1)c/2 + c/2 = (\ell + 1)c$.

Now observe that $R$ contains all $g(p) + (1, 0)$ for each $p \in P$ since $f(g(p) + (1, 0)) \ne f'(g(p) + (1, 0))$.
Therefore, shifting all vertices in $R$ by $(-1, 0)$, we obtain a Steiner tree of $g(P)$ with at most $(l+1)c$ vertices.
\end{proof}

% conclusion
\section{Concluding remarks}
\label{sec:conclusion}

We have studied \sts{} from the view point of restricted graph classes
and shown several positive and negative results.
We note that the complexity of the problem with the colorful condition remained unsettled for split graphs.
\begin{itemize}
  \item Is \sts{} NP-complete on split graphs in the colorful condition?
\end{itemize}

We did not address the approximability and the parameterized complexity of the problem in this paper,
which would be interesting research topics.

\paragraph{Approximation}
Since the problem is intractable in general, it would be interesting to ask if it admits an approximation.
It is known that for some constant $c$,
it is NP-hard to find a $c$-approximation solution for \sts~\cite{YamanakaDHKNOSS19}.
However, we can still hope for an approximation algorithm with a slightly worse approximation guarantee.
\begin{itemize}
  \item Does \sts{} admit a constant-factor approximation?
\end{itemize}
Note that the non-sequential variant, \textsc{Token Swapping},
admits a 4-approximation for general graphs~\cite{MiltzowNORTU16}.

\paragraph{Parameterized complexity}
It would be interesting to study the parameterized complexity of the problem (see~\cite{CyganFKLMPPS15}).
\Cref{lem:sub-pd-sum} and the $O(n^{3})$ upper bound of the minimum length of a swapping sequence~\cite{YamanakaDHKNOSS19}
together imply that \sts{} is fixed-parameter tractable parameterized by the maximum size of a biconnected component.
This parameter is an upper bound of \emph{treewidth}, the most well-studied structural graph parameter.\footnote{%
We omit formal definitions of the graph parameters mentioned in this section.
See e.g., \cite{HlinenyOSG08} for the their definitions and basic properties.}
Furthermore, one can see that block-cactus graphs have constant \emph{clique-width} (a generalization of treewidth).
Hence, it would be natural to ask the complexity of \sts{} parameterized by treewidth or clique-width,
although it looks quite challenging to generalize our algorithm to such settings.
Actually, we can observe that the problem is intractable when clique-width is the parameter.
It is easy to see that attaching a triangle to each vertex may increase clique-width only by a constant.
Since $(s,t)$-\textsc{Hamiltonian Path} is W[1]-hard parameterized by clique-width~\cite{FominGLS10},\footnote{%
Fomin et al.~\cite{FominGLS10} showed the W[1]-hardness of \textsc{Hamiltonian Cycle},
which can be easily translated to the W[1]-hardness of $(s,t)$-\textsc{Hamiltonian Path}
by adding a false-twin of an arbitrary vertex and call them $s$ and $t$.}
 the same argument in the proof of \cref{thm:few-color-gen} implies the following.
\begin{corollary}
\sts{} parameterized by clique-width is W[1]-hard. % in the few-color cases.
\end{corollary}

Given the discussion above, we would like to ask a few questions about the complexity of \sts{}.
\begin{itemize}
  \item Is it fixed-parameter tractable or XP parameterized by treewidth?
  How about upper bounds of treewidth such as vertex cover number?
  \item Is it XP parameterized by clique-width?
  Is it polynomial-time solvable for some graph classes of constant clique-width (e.g., cographs)?
\end{itemize}

% bib
\bibliographystyle{elsarticle-num}
\bibliography{ref}

\end{document}